\documentclass[11pt]{article}%
\usepackage{fullpage}
\usepackage{amssymb,amsmath}
\usepackage{graphicx, epsfig}
\usepackage{float}
\usepackage{bbm}
\def\idrm#1{\ensuremath{\mathrm{#1}}}

\def\floor#1{\lfloor #1 \rfloor}
\def\ceil#1{\lceil #1 \rceil}

\newcommand{\no}[1]{}

\newcommand{\todo}[1]{} %{{TO DO: {\sc #1}}}

\newtheorem{theorem}{Theorem}
\newtheorem{lemma}{Lemma}

\newtheorem{transform}{Transformation}

\newenvironment{proof}{\trivlist\item[]\emph{Proof}:}%
{\unskip\nobreak\hskip 1em plus 1fil\nobreak$\Box$
\parfillskip=0pt%
\endtrivlist}

\newenvironment{itemize*}%
  {\begin{itemize}%
    \setlength{\itemsep}{0pt}%
    \setlength{\parskip}{0pt}%
    \setlength{\parsep}{0pt}%
    \setlength{\topsep}{0pt}%
    \setlength{\partopsep}{0pt}%
  }%
  {\end{itemize}}%

\newcommand{\cT}{{\cal T}}

\newcommand{\cD}{{\cal D}}

\newcommand{\cC}{{\cal C}}
\newcommand{\cI}{{\cal I}}
\newcommand{\cL}{{\cal L}}
\newcommand{\cN}{{\cal N}}

\newcommand{\Temp}{{Temp}}
\newcommand{\bT}{\mathbb{T}}
\newcommand{\bN}{\mathbb{N}}
\newcommand{\bC}{\mathbb{C}}
\newcommand{\bL}{\mathbb{L}}
\newcommand{\bR}{\mathbf{R}}
\newcommand{\bA}{\mathbf{A}}

\newcommand{\of}{\overline{f}}
\newcommand{\on}{\overline{n}}

\newcommand{\eps}{\varepsilon}

\newcommand{\occ}{\mathrm{occ}}

\newcommand{\tcount}{t_{\mathrm{count}}}
\newcommand{\trange}{t_{\mathrm{range}}}
\newcommand{\tlocate}{t_{\mathrm{locate}}}
\newcommand{\textract}{t_{\mathrm{extract}}}
\newcommand{\tsa}{t_{\mathrm{SA}}}
\newcommand{\ra}{\idrm{rank}}
\newcommand{\sel}{\idrm{select}}

\newcommand{\shortver}[1]{}
\newcommand{\longver}[1]{#1}
\newcommand{\shlongver}[2]{#2}

\begin{document}
\title{
Dynamic Data Structures for Document Collections and Graphs}
\author{
J. Ian Munro\thanks{Cheriton School of Computer Science, University of Waterloo. Email {\tt imunro@uwaterloo.ca}.}
\and
Yakov Nekrich\thanks{Cheriton School of Computer Science, University of Waterloo.
Email: {\tt yakov.nekrich@googlemail.com}.}
\and
Jeffrey Scott Vitter\thanks{Department of Electrical Engineering \& Computer Science, University of Kansas.
Email: {\tt jsv@ku.edu}.}
}
\date{}
\maketitle

% \numberofauthors{3}
% \author{
% \alignauthor
% J. Ian Munro\\
% \affaddr{Cheriton School of Computer Science}\\
% \affaddr{University of Waterloo} \\
% \email{imunro@uwaterloo.ca}.
% \alignauthor
%  Yakov Nekrich\\
%    \affaddr{Cheriton School of Computer Science}\\
%    \affaddr{ University of Waterloo}\\
%    \email{ ynekrich@uwaterloo.ca}
% \alignauthor
% Jeffrey Scott Vitter\\
% \affaddr{Department of Electrical Engineering \& Computer Science}\\
% \affaddr{University of Kansas}\\
%   \email{jsv@ku.edu}.
% }

% \maketitle

%\fontsize{10pt}{10.2pt}
%\selectfont

%\thispagestyle{empty}
\begin{abstract}
% In this paper we present a general framework for adding dynamism to compressed data structures and show how it can be applied  to the problem of indexing a dynamic document collection.
In the dynamic indexing problem, we must maintain a changing collection of text documents so that we can efficiently 
support insertions, deletions, and pattern matching queries.
We are especially interested in developing efficient data structures that store and query the documents in compressed form.
All previous compressed solutions to this problem rely on answering rank and select queries on a dynamic sequence of symbols. 
Because of the lower bound in [Fredman and Saks, 1989], 
 answering rank queries presents a bottleneck in compressed dynamic indexing. In this paper we show how this lower bound can be circumvented using our new framework.
We demonstrate that the gap between static and dynamic variants of the indexing problem can be almost closed. Our method is based on a novel framework for adding dynamism to static compressed data structures.
Our framework also applies more generally to dynamizing other problems.
We show, for example,  how our framework can be applied to develop compressed representations of dynamic graphs and binary
 relations. 
\end{abstract}
%\category{F2.2}{Analysis of Algorithms and Problem Complexity}{Nonnumerical Algorithms and Problems}[Geometrical problems and computations]
%A category including the fourth, optional field follows...
%\category{H.3.1}{Information Storage and Retrieval}{Content Analysis and Indexing}[Indexing methods]


%\keywords{Compressed Data Structures, Text Indexes, Graph Data Structures}

\pagestyle{plain}
\section{Introduction}
Motivated by the preponderance of massive data sets (so-called ``big data''), it is becoming increasingly useful to store data in compressed form 
and moreover to  manipulate and query the data while in compressed form.  For that reason, such compressed data structures have 
been developed in the context of text indexing, graph representations, XML indexes, labeled trees, and many other applications.  
In this paper we describe a general framework to convert known static compressed data structures into dynamic  compressed data structures. We show how this framework can be used to obtain significant improvements for two important dynamic problems: maintaining a  dynamic graph and storing a dynamic collection of documents. We expect that our general framework will find further applications.

In the indexing problem, we keep a text or a collection of texts in a data structure, so that, given a query pattern, we can list 
all occurrences of the query pattern in the texts. This problem is one of the most fundamental 
in the area of string algorithms. Data structures that use 
$O(n\log n)$ bits of space can answer pattern matching queries 
in optimal time $O(|P|+\occ)$, where $|P|$ denotes the length 
of the query pattern $P$ and $\occ$ is the number of occurrences of~$P$. 
Because of the large volumes of data stored in text data bases 
and document collections, we are especially interested in data structures that store the text or texts in compressed form and at the same time can answer pattern matching queries efficiently. 
Compressed indexing problem was extensively studied in the static scenario and during the last two decades significant 
progress has been achieved; we refer to a survey~\cite{NavarroM07} for an overview of previous results in this area. 

In the dynamic indexing problem, also known as the library management problem, we maintain a collection of documents (texts) in a data structure under insertions and deletions of texts. It is not difficult to keep a dynamic collection of texts in~$O(n)$ words (i.e., $O(n\log n)$ bits) and support pattern matching queries at the same time. For instance, we can maintain suffixes of all texts in a suffix tree; when a new text is added or deleted, we add  all suffixes of the new text  to the suffix tree (respectively, remove all suffixes of the deleted text from the suffix tree). We refer the reader \shlongver{to the full version}{Section~\ref{sec:uncompr}} for a more detailed description of the $O(n\log n)$-bit solution. The problem of keeping a dynamic document collection in compressed form is however more challenging. Compressed data structures for the library management problem were considered in a number of papers~\cite{CHL04,MN06,CHLS07,MN06,LP07,MN08,GN08,LP09,GN09,HM10,NS10,NavarroN13}. In spite of previous work,  the query times of previously described dynamic data structures  significantly exceed the query times of the best static indexes. %by almost a logarithmic factor.
In this paper we show that the gap between the static and the dynamic variants of the compressed indexing problem can be closed or almost closed. 
Furthermore we show that our approach can be applied to the succinct representation of dynamic graphs and binary relations that supports basic adjacency and neighbor queries. Again 
our technique significantly reduces the gap between static and dynamic variants of this problem.

These problems arise often in database applications. 
For example, reporting or counting occurrences of a string in a dynamic collection of documents is an important operation in text databases and web browsers.
Similar tasks also arise in data analytics. Suppose that we keep a search log and want to find out how many times URLs containing a certain substring were accessed. 
Finally the indexing problem is closely related to the problem of substring occurrence estimation~\cite{OrlandiV11}. The latter problem is used in solutions of the substring selectivity estimation problem~\cite{ChaudhuriGG04,JagadishNS99,KrishnanVI96}; we refer to~\cite{OrlandiV11}  for a more extensive description. 
 Compressed storage schemes for such problems  help us save space and boost general performance because a larger portion of  data can reside in the fast memory.
Graph representation of data is gaining importance in the database community. For instance, the set of subject-predicate-object RDF triples can be represented as a graph or as  two binary relations~\cite{FernandezMG10}.  
Our compressed representation applied to an
RDF graph enables us to support basic reporting and
counting queries on triples. An example of such a query is given~$x$, to
enumerate all the triples in which $x$ occurs as a subject.
Another example is, given $x$ and $p$, to enumerate all triples in which $x$ occurs as a subject and $p$ occurs as a predicate.
\paragraph{Previous Results. Static Case}
 We will denote by $|T|$ the number of symbols in a sequence $T$ or in a collection of sequences; $T[i]$ denotes the $i$-th element in a sequence $T$ and $T[i..j]=T[i]T[i+1]\ldots T[j]$.  
Suffix trees and suffix arrays are two handbook data structures for the indexing problem.
Suffix array  keeps (references to) all suffixes $T[i..n]$ of  a text $T$ in lexicographic order. Using a suffix array, we can find the range of suffixes starting with a query string $P$ in $\trange=O(|P|+\log n)$ time; once this range is found, we can locate each occurrence of~$P$ in~$T$ in $\tlocate=O(1)$ time. 
A suffix tree is a compact trie that contains references to all suffixes $T[i..n]$ of a text $T$.
Using a suffix trie, we can find the range of suffixes starting with a query string $P$ in $\trange=O(|P|)$ time; once this range is found, we can locate every occurrence of~$P$ in~$T$ in $\tlocate=O(1)$ time. 
A large number of compressed indexing data structures are described in the literature; we refer to~\cite{NavarroM07} for a survey. These data structures follow the same two-step procedure for answering a query: first, the range  of suffixes that start with $P$ is found in $O(\trange)$ time, then we locate each occurrence of~$P$ in~$T$ in $O(\tlocate)$ time. Thus we report all $\occ$ occurrences of~$P$  in $O(\trange +\occ\cdot\tlocate)$ time. We can also extract any substring $T[i..i+\ell]$ of~$T$ in $O(\textract)$ time. Data structures supporting queries on a text $T$ can be extended to answer queries on a collection $\cC$ of texts: it suffices to append a unique symbol $\$_i$ at the end of every text $T_i$ from $\cC$ and keep the concatenation of all $T_i$ in the data structure. 

We list the currently best and selected previous results for static text indexes with asymptotically optimal space usage in Table~\ref{tbl:best1}. All listed data structures can achieve different space-time trade-offs that depend on parameter $s$: an index typically needs about
$nH_k+o(n\log\sigma)+O(n\log n/s)$ bits and $\tlocate$ is proportional to $s$.
Henceforth $H_k$ denotes the $k$-th order empirical entropy and $\sigma$ denotes the alphabet size\footnote{Let $S$ be an arbitrary string over an alphabet $\Sigma=\{\,1,\ldots,\sigma\,\}$. A \emph{context} $s_i\in \Sigma^k$ is an arbitrary string of length $k$.  Let $n_{s_i,a}$ be the number of times the symbol $a$ is preceded by a context $s_i$ in $S$ and $n_{s_i}=\sum_{a\in \Sigma}n_{s_i,a}$. Then $H_k=-\sum_{s_i\in \Sigma^k} \sum_{a\in \Sigma}n_{s_i,a}\log \frac{n_{s_i,a}}{n_{s_i}}$ is the $k$-th order empirical entropy of $S$.}. We assume that $k\le \alpha\log_{\sigma}n-1$ for a constant $0<\alpha<1$. $H_k$ is the lower bound on the average space usage of any statistical compression method that encodes each symbol using the context of~$k$ previous symbols~\cite{Manzini01}. 
 The currently fastest such index of Belazzougui and Navarro~\cite{BelazzouguiN11} reports all occurrences of~$P$ in $O(|P|+ s\cdot\occ)$ time and extracts a substring of length $\ell$ in 
$O(s+\ell)$ time.
Thus their query time depends only on the parameter $s$ and the length of~$P$. Some recently described indices~\cite{BGNN10,BelazzouguiN11}  achieve space usage $nH_k+o(nH_k)+o(n)$ or $nH_k+o(nH_k)+O(n)$ instead of~$nH_k+o(n\log \sigma)$.

If we are interested in obtaining faster data structures and can use $\Theta(n\log\sigma)$ bits of space, then better trade-offs between space usage and time are possible~\cite{GGV03,GrossiV05}. 
For the sake of space, we describe only one such result.
The data structure of Grossi and Vitter~\cite{GrossiV05} uses $O(n\log \sigma)$ bits and reports occurrences of a pattern 
in $O(|P|/\log_{\sigma}n+\log^{\eps}n+ \occ\log^{\eps}n)$ time; see Table~\ref{tbl:nlogsigma}.
We remark that the  fastest data structure in Table~\ref{tbl:best1} needs 
$\Omega(n\log^{1-\eps}n)$ space to obtain the same time 
for $\tlocate$ as in~\cite{GrossiV05}. If a data structure from Table~\ref{tbl:best1} uses 
$O(n\log\sigma)$ space, then $\tlocate=\Omega(\log_{\sigma} n)$.

\tolerance=1000
%{\bf Dynamic Indexing.}
% In the dynamic indexing problem, also known as the library management problem, we maintain a collection of documents (texts) in a data structure under insertions and deletions of texts. It is not difficult to keep a dynamic collection of texts in $O(n)$ words  and support pattern matching queries at the same time. For instance, we can add or remove all suffixes of a text  to a suffix tree. The problem of keeping a dynamic document collection in compressed form is however more challenging. Compressed data structures for the library management problem were considered in a number of papers
\paragraph{Dynamic Document Collections}
In the dynamic indexing problem, we maintain a collection of documents (strings) under insertions and deletions. An insertion adds a new document to the collection, 
a deletion removes a document from the collection. For any query substring $P$, we must return all occurrences of $P$ in all documents. When a query is answered, relative positions of occurrences are reported. To be precise, we must report all pairs $(\mathit{doc},\mathit{off})$, such that $P$ occurs in a document $\mathit{doc}$ at position $\mathit{off}$. 
We remark that  relative positions of $P$ (with respect to document boundaries) are reported. Hence an insertion or a deletion of a document does not change positions of $P$ in other documents. Indexes for dynamic collections of strings were also studied extensively~\cite{CHL04,MN06,CHLS07,LP07,MN08,GN08,LP09,GN09,HM10,NS10,NavarroN13}. The fastest previously known result for the case of large alphabets is described in~\cite{NavarroN13}.  
Their data structure, that builds on a long line of previous work,  uses $nH_k+o(n\log \sigma) + O(n\log n/s) +O(\rho\log n)$ bits of space, where  $\rho$ is the number of documents; queries are answered  in $O(|P|\log n/\log\log n + \occ\cdot s\cdot \log n /\log \log n)$ time and updates are supported in $O(\log n + |T_u|\log n/\log \log n)$ amortized time, where $T$ is the document inserted into or deleted from the collection. See Table~\ref{tbl:best2} for some other previous results.
\begin{table*}[tb]
  \centering
\resizebox{\textwidth}{!}{
  \begin{tabular}{|c|l|l|l|l|l|} \hline
    Ref. & Space ($+O(n\frac{\log n}{s})$)& $\trange$ & $\tlocate$ & $\textract$ &  $\sigma$ \\
         &                                &           &            &             &            \\ \hline
\cite{GGV03} & $nH_k+o(n\log\sigma)$ & $O(|P|\log\sigma+\log^4n)$ & $O(s\log\sigma)$ & $O((s+\ell)\log\sigma)$  & \\
\cite{Sadakane03}    & $nH_k + o(n\log \sigma)$ & $O(|P|\log n)$ & $O(s)$ & $O(s+\ell)$  & \\
\cite{FMMN07} & $nH_k + o(n\log \sigma)$ & $O(|P| \frac{\log\sigma}{\log\log n})$ & $O(s\frac{\log\sigma}{\log \log n})$ & $O((s+\ell)\frac{\log \sigma}{\log\log n})$  &  \\
\cite{BHMR07} & $nH_k + o(n\log \sigma)$ &  $O(|P| \log\log \sigma)$ & $O(s\log\log\sigma)$ & $O((s+\ell)\log\log \sigma)$  &  \\
\cite{BGNN10} & $nH_k +o(nH_k)+ o(n)$ & $O(|P| \frac{\log\sigma}{\log\log n})$ & $O(s\frac{\log\sigma}{\log \log n})$ & $O((s+\ell)\frac{\log \sigma}{\log\log n})$  &  \\
\cite{BGNN10} & $nH_k +o(nH_k)+ o(n)$ & $O(|P| \log\log \sigma)$ & $O(s\log\log\sigma)$ & $O((s+\ell)\log\log \sigma)$  &  \\
\cite{BGNN10} & $nH_k +o(nH_k)+ o(n)$ & $O(|P|)$ & $O(s)$ & $O(s+\ell)$  & $\log^{\mathrm{const}}n$ \\
\cite{BelazzouguiN11} & $nH_k+o(nH_k)+O(n)$ & $O(|P|)$ & $O(s)$ & $O(s+\ell)$  & \\ \hline   
  \end{tabular}
}
\caption{Asymptotically optimal space data structures for  static indexing. Occurrences of a string $P$ can be found in $O(\trange+\tlocate\cdot \occ)$ time. A substring $T[i..i+\ell]$ of~$T$ can be extracted in $O(\textract)$ time.
    Results are valid for any  $k\le \alpha\log_{\sigma}n-1$ and $0<\alpha<1$.}
\label{tbl:best1}
\end{table*}
%$~$\\[.32cm]
\begin{table*}[tb]
  \centering
\resizebox{\textwidth}{!}{
  \begin{tabular}{|c|l|l|l|l|l|l|} \hline
    Ref. & Space & $\trange$ & $\tlocate$ & $\textract$ & Insert/ & $\sigma$ \\
         & ($+O(n\frac{\log n}{s})+\rho\log n$)      &           &            &             & Delete  &           \\ \hline
\cite{CHLS07} & $O(n)$ & $O(|P|\log n)$ & $O(\log^2 n)$ & $O((\log n+\ell)\log n)$ & $O(|T_u|\log n)$ & $\mathrm{const}$ \\
\cite{MN06}   & $nH_k+o(n\log\sigma)$ & $O(|P|\log n \log \sigma)$ & $O(s \log n \log \sigma)$ & $((s+\ell)\log n \log \sigma)$ & $O(|T_u|\log n\log \sigma)$ &  \\
\cite{NavarroN13} & $nH_k+o(n\log\sigma)$ & $O(|P|\log n)$ & $O(s \log n)$ & $O((s+\ell)\log n)$ & $O(|T_u|\log n)$ &  \\
\cite{NavarroN13} & $nH_k+o(n\log\sigma)$ & $O(|P|\frac{\log n}{\log\log n})$ & $O(s \log n /\log\log n)$ & $O((s+\ell)\frac{\log n}{\log\log n})$ & $O(\log n + |T_u|\frac{\log n}{\log\log n})^{\bA}$ &  \\ \hline
Our & $nH_k+o(n\log\sigma)$ & $O(|P|\log\log n)$ & $O(s)$ & $O(s+\ell)$ & $O(|T_u|\log^{1+\eps} n)$ & $\log^{\mathrm{const}}n$ \\
Our & $nH_k+o(n\log\sigma)$ & $O(|P|\log\log n\log\log\sigma)$ & $O(s\log\log\sigma)$ & $O((s+\ell)\log\log\sigma)$ & $O(|T_u|\log^{\eps} n)$/ &  \\
    &                            &                    &        &             &  $O(|T_u|(\log^{\eps}n+s))$ \\ 
Our & $nH_k+o(n\log\sigma)$ & $O(|P|\log\log n)$ & $O(s)$ & $O(s+\ell)$ & $O(|T_u|\log^{\eps} n)^{\bR}$/ &  \\
    &                            &                    &        &             &  $O(|T_u|(\log^{\eps}n+s)^{\bR}$ \\ \hline
\end{tabular}
}
  \caption{Asymptotically optimal space data structures for  dynamic indexing. The same notation as in Table~\ref{tbl:best1} is used. 
Randomized update procedures that achieve specified cost in expectation are  marked with $\bR$.  Amortized update costs are marked with $\bA$. $T_u$ denotes the document that is inserted into (resp.\ deleted from) the data structure during an update operation.  In previous papers on dynamic indexing only the cases of~$s=\log n$ or $s=\log_{\sigma}n\log \log n$ was considered, but extension to an arbitrary value of~$s$ is straightforward.}
\label{tbl:best2}
\end{table*}

An important component of previous dynamic solutions  is a data structure supporting rank and select queries: a sequence $S$ over an alphabet $\Sigma=\{\,1,\ldots,\sigma\,\}$ is kept in a data structure so that the $i$-th occurrence of a
symbol $a\in \Sigma$ and the number of times a symbol $a$ occurs
in $S[1..i]$ for any $1\le i\le n$ can be computed.   
Thus progress in dynamic indexing was closely related to progress in dynamic data structures for rank 
and select queries.  In~\cite{NavarroN13} the authors obtain a dynamic data structure that supports rank and select in  $O(\log n/\log \log n)$ time. By the lower bound of Fredman and Saks~\cite{FS89}, this query time is optimal in the dynamic scenario. It was assumed that the solution of the library management problem described  in~\cite{NavarroN13} achieves query time that is close to optimal.

\paragraph{Our Results}
In this paper we show that the lower bound on dynamic rank-select problem can be circumvented and describe data structures that need significantly less time to answer queries.  Our results close or almost close the gap between static and dynamic  indexing.
If the alphabet size $\sigma=\log^{O(1)}n$, we can obtain an $(nH_k+o(n\log\sigma)  + O(n\frac{\log n}{s}))$-bit data structure that answers queries in $O(|P|\log\log n +  \occ \cdot s)$ time; updates are supported in  $O(|T_u|\log^{1+\eps}n)$  time, where $T_u$ denotes the document that is inserted into or deleted from the index. Our second data structure supports updates in $O(|T_u|\log^{\eps}n)$ expected time and answers queries in $O(|P|\log\log n +\occ\cdot s)$ time for an arbitrarily large alphabet\footnote{Dynamic indexes also need $O(\rho\log n)$ bits to navigate between documents, where $\rho$ is the number of documents. Since $\rho\log n$ is usually negligible in comparison to $n$, we ignore this additive term, except for Tables~\ref{tbl:best2} and ~\ref{tbl:nlogsigma}, to simplify the description.}. If the update procedure is deterministic, then queries are answered in $O((|P|\log\log n +\occ\cdot s)\log\log \sigma)$ time and updates are supported in 
$O(|T_u|\log^{\eps}n)$ worst-case time.  See Table~\ref{tbl:best2}. If  $O(n\log\sigma)$ bits of space are available, then our dynamic data structure matches the currently fastest static result of Grossi and Vitter~\cite{GrossiV05}. We can report all occurrences of 
a pattern $P$ in $O(|P|/\log_{\sigma}n + \log^{\eps}n + \occ\cdot\log^{\eps}n)$ time. 
This is the first compressed dynamic data structure that achieves $\trange=o(|P|)$ if $\sigma=n^{o(1)}$. Compared to the fastest previous  data structure that needs the same space, we achieve $O(\log n\log_{\sigma}n)$ factor improvement in query time.  A variant of this data structure with deterministic update procedure answers queries in $O(|P|(\log\log n)^2/\log_{\sigma}n + \log n + \occ\cdot\log^{\eps}n)$ time. See Table~\ref{tbl:nlogsigma}.

Our data structures can also count occurrences of a pattern $P$ 
in $O(\tcount)$ time. For previously described indexes $\tcount=\trange$. In our case, $\tcount=\trange+\log n/\log \log n$ or $\tcount=(\trange+\log n/\log\log n)\log\log n$. 
Times needed to answer a counting query are listed in  Table~\ref{tbl:count}.
However, if our data structures support counting queries, then update times grow slightly, as shown in Table~\ref{tbl:count}.

%This is an almost $O(\log n)$ factor improvements over previous results. Other trade-offs between query and update costs are also possible; we refer to section~\ref{sec:dyncoll} for details.

All of the above mentioned results are obtained as corollaries of two general transformations.
Using these transformations, that work for a very broad class of indexes, we can immediately turn almost any  static data structure with good pre-processing time into an index for a dynamic collection of texts. 
The query time  either remains the same or increases by a very small multiplicative factor. 
Our method can be applied to other problems where both compressed representation  and dynamism are desirable. 

\paragraph{Binary Relations and Graphs}
One important area where our techniques can  also be used is compact representation of directed graphs and binary relations. Let $R\subseteq L\times O$ be a binary relation between labels from a set $L$ and objects from a set $O$. Barbay et al.~\cite{BHMR07}
describe a compact representation of a static binary relation $R$ (i.e., the set of object-label pairs) that consists of a sequence $S_R$ and a bit sequence $B_R$. $S_R$  contain the list of labels related to different objects and is ordered by object. That is, $S_R$ lists all labels related to an object $o_1$, then all labels related to an object $o_2$, etc. The binary sequence $B_R$ contains unary-encoded numbers of labels related to objects $o_1$, $o_2$, $\ldots$. Barbay et al~\cite{BHMR07} showed how $S_R$ and $B_R$ can be used to support basic queries on binary relations, such as listing or counting all labels related to an object, listing or counting all objects related to a label, and telling whether a label and an object are related. Their method reduces queries on a binary relation $R$ to rank, select, and access queries on $S_R$ and $B_R$. Another data structure that stores a static binary relation and uses the same technique is described in~\cite{BarbayCGNN14}. Static compact data structures described in~\cite{BHMR07,BarbayCGNN14} support queries in $O(\log\log \sigma_l)$ time per reported datum, where $\sigma_l$ is the number of labels. For instance, we can report  labels related to an object (resp. objects related to a label) in $O((k+1)\log\log\sigma_l)$ time, where $k$ is the number of reported items; we can tell whether an object and a label are related in $O(\log\log \sigma_l)$ time. 
In~\cite{NavarroN13}, the authors describe a dynamization of this approach that relies on  dynamic data structures answering rank and select queries on dynamic strings $S_R$ and $B_R$. Again the lower bound on dynamic rank  queries sets the limit on the efficiency of this approach. Since we need $\Omega(\log n/\log \log n)$ time to answer rank queries, the data structure of Navarro and Nekrich~\cite{NavarroN13} needs $O(\log n/\log\log n)$ 
time per reported item, where $n$ is the number of object-label pairs in~$R$. Updates are supported in $O(\log n/\log\log n)$ amortized time and the space usage is $nH+\sigma_l\log\sigma_l + t\log t +O(n+\sigma_l\log \sigma_l)$ where $n$ is the number of pairs, $H$ is the zero-order entropy of the string $S_R$, $\sigma_l$ is the number of labels and $t$ is the number of objects.In~\cite{NavarroN13} the authors also show that we can answer basic adjacency and neighbor queries on a directed graph by regarding a
graph as a binary relation between nodes. Again reporting and counting out-going and in-going neighbors of a node can be performed in $O(\log n)$ time per delivered datum.

In this paper we show how our method for dynamizing compressed data structures can be applied to binary relations and graphs. Our data structure supports reporting labels related to an object or reporting objects related to a label in $O(\log \log n\cdot \log\log \sigma_l)$ time per reported datum.  We support counting queries in $O(\log n)$ time and updates in $O(\log^{\eps}n)$ worst-case time. The same query times are also achieved for the dynamic graph representation. The space usage of our data structures is dominated by $nH$ where $n$ is the number of pairs in a binary relation or the number of edges in a graph and $H$ is the zero-order entropy of the string $S_R$ defined above. Thus the space usage of our data structure matches that of~\cite{NavarroN13} up to lower-order factors. At the same time we show that reporting queries in a dynamic graph can be supported without dynamic rank and select queries. 

\begin{table*}[tb]
\centering
\resizebox{.8\textwidth}{!}{
  \begin{tabular}{|c|l|l|l|l|l|l|} \hline
    Ref. &  $\trange$ & $\tlocate$ & $\textract$ & Update & $\sigma$ \\
         &           &            &             &   &           \\ \hline
\cite{GrossiV05}  & $O(|P|/\log_{\sigma} n+\log^{\eps}n)$ & $O(\log^{\eps} n)$ & $O(\ell/\log_{\sigma}n)$ & static &  \\
\cite{CHLS07} &  $O(|P|\log n)$ & $O(\log^2 n)$ & $O((\log n+\ell)\log n)$ & $O(|T_u|\log n)$ & $\mathrm{const}$ \\
\cite{NavarroN13} & $O(|P|\log n)$ & $O(\log n\log_{\sigma}n)$ & 
$O((\log_{\sigma} n + \ell)\log n)$ & $O(|T_u|\log n)$ & \\
Our & $O(|P|/\log_{\sigma} n+\log^{\eps}n)$ & $O(\log^{\eps} n)$ & $O(\ell/\log_{\sigma}n)$ & $O(|T_u|\log^{\eps}n)^{\bR}$ &  \\ 
Our & $O(|P|(\log\log n)^2/\log_{\sigma} n+\log n)$ & $O(\log^{\eps} n)$ & $O(\ell/\log_{\sigma}n)$ & $O(|T_u|\log^{\eps}n)$ &  \\ 

\hline
\end{tabular}
}
\caption{$O(n\log\sigma)$-bit indexes. Dynamic data structures need additional $\rho\log n$ bits. Randomized update costs are  marked with $\bR$. }
\label{tbl:nlogsigma}
\end{table*}

\begin{table*}[tbh]
\centering
\resizebox{.8\textwidth}{!}{
  \begin{tabular}{|c|l|l|l|l|l|l|} \hline
   Space  & Counting & Updates & $\sigma$ \\ \hline
 $nH_k+o(n\log\sigma)$ & $O(|P|\log\log n +\log n)$  &  $O(|T_u|\log n)$ &  $\log^{\mathrm{const}}n$\\
$nH_k+o(n\log\sigma)$ & $O((|P|\log\log\sigma\log\log n +\log n)$ & $O(|T_u|\log n)$ &  \\ 
$nH_k+o(n\log\sigma)$ & $O((|P|\log\log n +\log n)$ & $O(|T_u|\log n)^{\bR}$ &  \\ 
$O(n\log\sigma)$ & $O(|P|/\log_{\sigma} n+\log n/\log\log n)$ &  $O(|T_u|\log n)^{\bR}$ &  \\ 
$O(n\log\sigma)$ &  $O(|P|(\log\log n)^2/\log_{\sigma} n+\log n)$ &
$O(|T_u|\log n)$ &  \\ \hline
  \end{tabular}
}
\caption{Costs of counting queries for our data structures. Randomized update costs are  marked with $\bR$. The first three rows correspond to the last three rows in Table~\ref{tbl:best2}, the last two rows correspond to the last two rows in Table~\ref{tbl:nlogsigma}.}
\label{tbl:count}
\end{table*}

\paragraph{Overview}
The main idea of our approach can be described as follows. 
The input data is distributed among several data structures. 
We maintain a fraction of the data in an uncompressed data structure that supports both insertions and deletions. 
We bound the number of elements stored in uncompressed form 
so that the total space usage of the uncompressed data structure 
is affordable.
Remaining data is kept in several compressed data structures that do not support updates.  
New elements (respectively new documents) are always inserted 
into the uncompressed data structure. Deletions from the 
static data structures are implemented by the lazy deletions mechanism: when a deletion takes place, then the deleted  element (respectively the document) is marked as deleted. We keep positions of marked elements in a data structure, so that all elements in a query range that are not marked as deleted can be reported in $O(1)$ time per element. When a static structure contains too much obsolete data (because a certain fraction of its size is marked as deleted), then this data structure is purged: we create a new instance of this data structure that does not contain deleted elements. If the 
uncompressed data structure becomes too big, we  move its content  into  a (new) compressed data structure. 
Organization of compressed data structures is inspired by 
the logarithmic method, introduced by Bentley and Saxe~\cite{BentleyS80}: the size of compressed data structures increases geometrically. 
We show that re-building procedures can be scheduled in such 
way that only a small fraction of data is kept in uncompressed form at any given time. Since the bulk of data is kept in  static data structures, our approach can be viewed as  a general framework that transforms static compressed data structures into dynamic ones.

In Section~\ref{sec:dyncoll} we describe Transformation~\ref{trans:trans1}; Transformation~\ref{trans:trans1}, based on the approach outlined above,  can be used to turn a static indexing data structure  into a data structure for dynamic collection of documents with amortized update cost. The query costs of the obtained dynamic data structure are the same as in the underlying static data structure. %we show how a static indexing data
%structure can be transformed in order to solve the dynamic indexing problem with amortized updates. 
In Section~\ref{sec:worst} we describe Transformation~\ref{trans:trans1worst}
that turns a static indexing data structure into a dynamic data structure with worst-case update costs. 
We use more sophisticated division into sub-collections and slightly different re-building procedures in Transformation~\ref{trans:trans1worst}. 
In Section~\ref{sec:dynind} we describe how to obtain new solutions of the dynamic indexing problem using our static-to-dynamic transformations. Finally Section~\ref{sec:dyngraph} contains our data structures for dynamic graphs and binary relations.

\section{Dynamic Document Collections}\label{sec:dyncoll}
In this section we show how a static compressed index $\cI_s$ can be transformed into a dynamic index $\cI_d$.
 $\cC$ will denote a collection of texts $T_1,\ldots, T_{\rho}$. We say that an index $\cI_s$ is $(u(n), w(n))$-constructible if there is an algorithm that uses $O(n\cdot w(n))$ additional workspace and constructs $\cI_s$ in $O(n\cdot u(n))$ time. Henceforth we make the following important assumptions about the static index $\cI_s$.  $\cI_s$ needs at most $|S|\phi(S)$ bits of space for any symbol sequence $S$ and the function $\phi(\cdot)$ is monotonous: if any sequence $S$ is a concatenation of~$S_1$ and $S_2$, then $|S|\phi(S)\ge |S_1|\phi(S_1)+|S_2|\phi(S_2)$. 
We also assume that $\cI_s$ reports occurrences of a substring in $\cC$ using the two-step method described in the introduction:
first we identify the range $[a,b]$ in the suffix array, such that all suffixes that start with $P$ are in $[a,b]$; then we find the positions of suffixes from $[a,b]$ in the document(s).  These operations will be called range-finding and locating. 
Moreover the rank of any suffix $T_i[l..]$ in the  suffix array can be found in time $O(\tsa)$.
The class of indexes that satisfy these conditions includes all indexes that are based on compressed suffix arrays or the Burrows-Wheeler transform. Thus the best currently known static indexes can be used in Transformation~\ref{trans:trans1} and the following transformations described in this paper. 

\tolerance=1000
 Our result can be stated as follows.
\begin{transform}
\label{trans:trans1}
Suppose that there exists a static $(u(n), w(n))$-constructible index $\cI_s$ that uses $|S|\phi(S)$ space for any document collection $S$. Then there exists a dynamic index $\cI_d$ that uses $|S|\phi(S)+O(|S|(\frac{\log \sigma}{\tau}+w(n)+\frac{\log\tau}{\tau}))$ space for a parameter $\tau=O(\log n/\log \log n)$; $\cI_d$ supports insertions and deletions of documents in  $O(u(n)\log^{\eps}n)$ time per symbol and $O(u(n)\cdot \tau+\tsa+\log^{\eps}n)$ time per symbol respectively. Update times are amortized. The asymptotic costs of range-finding, extracting, and locating are the same in $\cI_s$ and $\cI_d$. 
\end{transform}
We start by  showing how to turn  a static index into a semi-dynamic deletion-only index using $O((n/\tau)\log\tau)$ additional bits.  Then we will show how to turn a semi-dynamic index into a fully-dynamic one. 

\paragraph{Supporting Document Deletions}
We keep a bit array $B$ whose entries correspond to positions 
in the suffix array $SA$ of~$\cC$. $B[j]=0$ if $SA[j]$ is a suffix of some text $T_f$, such  that $T_f$ was already deleted from $\cC$ and $B[j]=1$ otherwise.  We keep a data 
structure $V$ that supports the following operations on $B$: $zero(j)$ sets the $j$-th bit in~$B$ to~$0$; $report(j_1,j_2)$ reports all $1$-bits in $B[j_1..j_2]$.  $V$ is implemented using Lemma~\ref{lemma:fastbit2}, so that  $zero(i)$ is supported in $O(\log^{\eps}n)$ time and  $report(j_1,j_2)$ is answered in $O(k)$ time, where $k$ is the number of output bit positions.  If $B$ contains at most $n/\tau$ zeros, then $B$ and $V$ need only $O((n\log \tau)/\tau)$ bits. Lemma~\ref{lemma:fastbit2} is proved in Section~\ref{sec:fastbit}. 

When a document $T_f$ is deleted, we identify the positions of~$T_f$'s suffixes in~$SA$ and set the corresponding bits in~$B$ to~$0$. When the number of symbols in deleted documents equals $(n/\tau)$, we re-build the index for $\cC$ without deleted documents in $O(n\cdot u(n))$ time. The total amortized cost of deleting a document is $O(u(n)\tau +\tsa+\log^{\eps}n)$ per symbol. To report occurrences of 
some string $P$ in $\cC$, we identify the range $[s..e]$ such that all suffixes in $SA[s..e]$ start with $P$ in $O(\trange)$ time. Using $V$, we 
enumerate all $j$, such that $s \le j\le e$ and $B[j]=1$. 
For every such $j$, we compute $SA[j]$ in $O(\tlocate)$ time.

\paragraph{Fully-Dynamic Index}
We split $\cC$ into a constant number of sub-collections $\cC_0, \cC_1,\ldots, \cC_r$ such that 
$|\cC_i|\le \max_i$ for all $i$. The maximum size of the $i$-th sub-collection, $\max_i$, increases geometrically: $\max_0=2n/\log^2n$ and $\max_i= 2(n/\log^2n)\log^{\eps\cdot i}n$ for a constant $\eps>0$; see Fig.~\ref{fig:cascad1}.
There is no lower bound on the number of symbols in a sub-collection $\cC_i$; for instance, any $\cC_i$ can be empty.
Our main idea is to store  $\cC_0$ in uncompressed form  and $\cC_i$ for $i\ge 1$  in semi-static deletion-only data structures.
Insertions into $\cC_i$ for $i\ge 1$ are supported by re-building the semi-static index of~$\cC_i$. 
We also re-build all sub-collections when the total number of elements is increased by a constant factor (global re-build).
\todo{should I use $size(i)$ or $|\cC_i|$ to define $\max_i$?}
\begin{figure}[tb]
  \centering
  \includegraphics[height=.2\textheight]{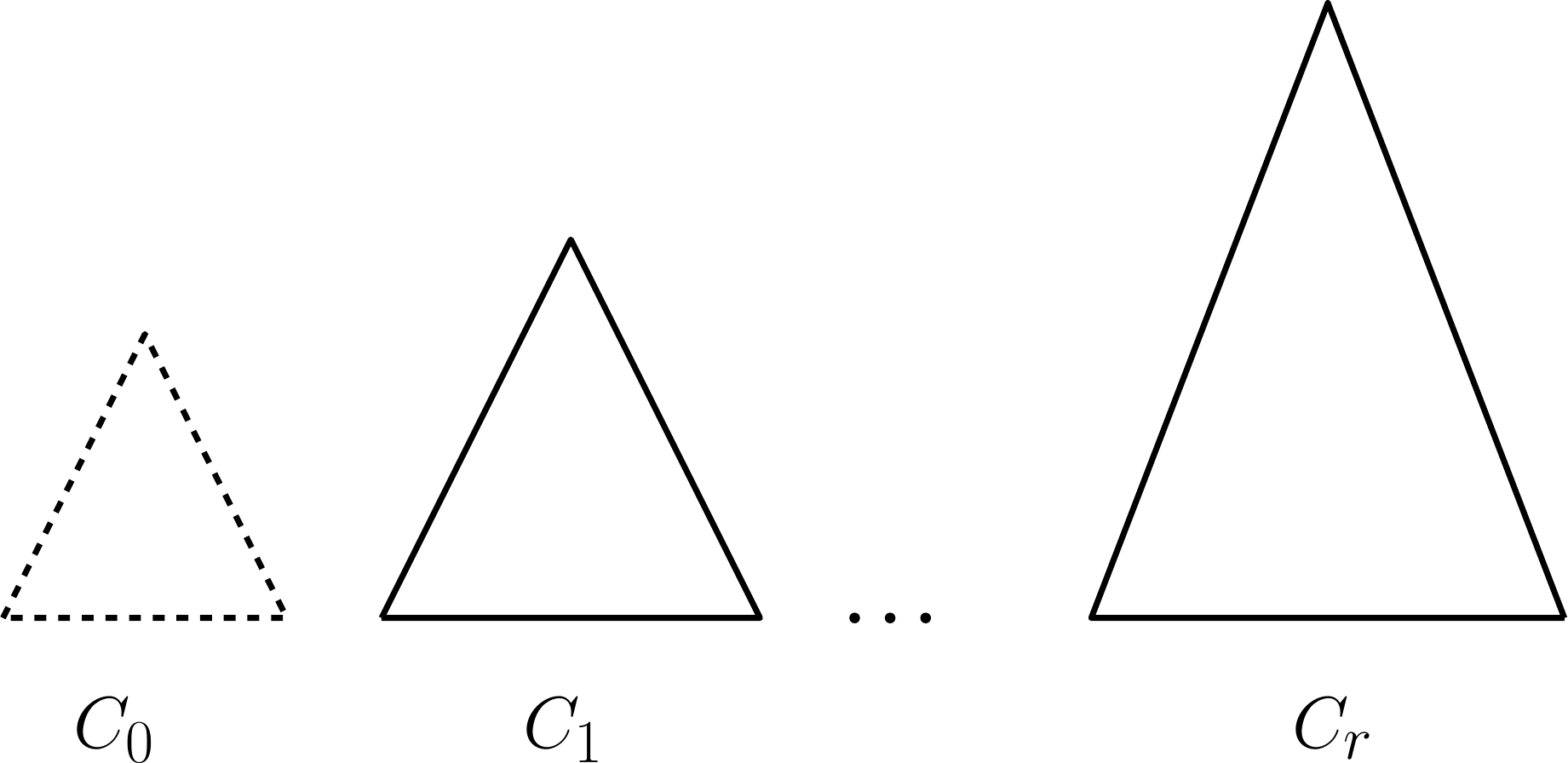}
  \caption{Sub-collections $\cC_i$ for dynamizing  a deletion-only index. A data structure for $\cC_0$ is fully-dynamic and stores documents in uncompressed form.}
  \label{fig:cascad1}
\end{figure}
  
We store the document collection $\cC_0$ in uncompressed form. 
Suffixes of all documents in $\cC_0$ are kept in an (uncompressed) suffix tree $\cD_0$.  We can insert a new text $T$ into $\cD_0$ or delete $T$ from $\cD_0$ 
in $O(|T|)$ time. Using $\cD_0$, all occurrences of a pattern $P$ in~$\cC_0$ can be reported in $O(|P|+\occ)$ time. 
Since $|\cC_0|\le 2n/\log^2n$, we need $O(n/\log n)$ bits to store $\cC_0$.  For completeness we will describe the data structure for $\cC_0$ \shlongver{in the full version of this paper}{in Section~\ref{sec:uncompr}.} 

Every $\cC_i$ for $i\ge 1$ is kept in a semi-dynamic data structure described in the first part of this section.  Let $size(i)$ denote the total length of all undeleted texts in $\cC_i$. When a new document $T$ must be inserted, we find the first (smallest) collection $\cC_j$ such that $\sum_{i=0}^j size(i)+|T|\le \max_j$ where $\max_j=2(n/\log^2n)\log^{\eps\cdot j}n$. 
That is, we find the first subcollection $\cC_j$ that can accommodate the new text $T$ and all preceding subcollections without 
exceeding the size limit. If $j=0$, we insert the new text into $\cC_0$. Otherwise, if $j\ge1$, we discard the old indexes for all $\cC_i$ where $0\le i\le j$,  set $\cC_j=(\cup_{i=0}^j \cC_i)\cup T$ and construct a new semi-static index for $\cC_j$.  If $\sum_{i=0}^j size(i) +|T|> \max_j$ for all $j$, we start a global re-build procedure: all undeleted texts from old sub-collections are moved to the new sub-collection $\cC_r$ and parameters $\max_i$ are re-calculated; \no{used to determine maximal sizes of sub-collections} new sub-collections $\cC_i$ for $0\le i< r$ are initially empty after the global re-build. 

We start a global re-build procedure when the total number of elements is at least doubled. Hence, the amortized cost of a global re-build is $O(u(n))$. The amortized cost of re-building sub-collections can be analyzed as follows. 
When a sub-collection $\cC_j$ is re-built, we insert all symbols from subcollections $\cC_i$, $0\le i<j$ and the new text $T$ into $\cC_j$. Our insertion procedure guarantees that $\sum_{i=1}^{j-1}size(j)+|T|> \max_{j-1}$. 
We need $O(\max_j\cdot u(n))$ time to construct a new index for $\cC_j$. The cost of re-building $\cC_j$ can be distributed among the new text symbols inserted into $\cC_j$. Since $\max_{j-1}=\max_j/\log^{\eps}n$, the amortized cost of inserting a new symbol into $\cC_j$ is $O(u(n)\cdot\log^{\eps}n)$.
Every text is moved at most once  to any subcollection $\cC_j$ for any $j$ such that $1\le j \le \ceil{2/\eps}$.
Hence the total amortized cost of an insertion is $O((1/\eps)u(n)\cdot\log^{\eps}n)$  per symbol.    

A query is answered by querying all non-empty sub-collections $\cC_i$ for $i=0,1,\ldots, r$.
Since $r=O(1)$, query times are the same as in the underlying static index. 
Splitting a collection into sub-collection does not increase the space usage because the function $\phi(\cdot)$ is monotonous. We need $O((n/\tau)\log \tau)$ bits to keep data structures $V$ for all $\cC_i$. Another $O(nw(n))$ bits are needed for global and local re-builds. Finally we need $O((n/\tau)\log\sigma)$ bits to store the symbols from deleted documents. Since there are no more than $O(n/\tau)$ deleted symbols, we use $O((n/\tau)\log\sigma)+o(n\log\sigma)$ additional bits to store them; a more detailed analysis \shlongver{will be given in the full version.}{is given in Section~\ref{sec:spacean}.} Hence, the total space overhead of our dynamic index is $O(n(w(n)+(\log\sigma+\log\tau)/\tau))$.

\longver{A data structure with faster insertions and slightly higher query time can be obtained by increasing the number of sub-collections $\cC_i$ to $O(\log\log n)$.
We describe this variant of our method in Appendix~\ref{sec:fasteramortized}.}

\section{Worst-Case Updates}
\label{sec:worst}
In this section we will prove the following result.
\begin{transform}
\label{trans:trans1worst}
Suppose that there exists a static $(u(n), w(n))$-constructible index $\cI_s$ that uses $|S|\phi(S)$ space for any document collection $S$. Then there exists a dynamic index $\cI_d$ that uses $|S|\phi(S)+O(|S|\frac{\log\sigma+\log\tau+w(n)}{\tau})$ space for any parameter $\tau=O(\log n/\log \log n)$; $\cI_d$ supports insertions and deletions of documents in  $O(u(n)\log^{\eps}n)$ time per symbol and $O(u(n)\cdot (\log^{\eps}n+\tau \log \tau)+\tsa)$ time per symbol respectively. The asymptotic costs of range-finding increases by $O(\tau)$; the costs of  extracting and locating are the same in $\cI_s$ and $\cI_d$. 
\end{transform}
 We use the index of Transformation~\ref{trans:trans1} as the starting point.
First we give an overview of our data structure and show how queries can be answered.  Then we describe the procedures for text  insertions and deletions. 
\paragraph{Overview}
The main idea of supporting updates in worst-case is to maintain several copies of the same sub-collection. An old copy of~$\cC_j$ is locked while a new updated version of~$\cC_{j+1}$ that includes $\cC_j$ is created in the background. When a new version of~$\cC_{j+1}$ is finished, we discard an  old locked sub-collection. 
When a new document $\cT$ must be inserted, we insert it into $\cC_0$ if $|\cC_0|+|T|\le \max_0$.
Otherwise we look for the smallest $j\ge 0$, such that $\cC_{j+1}$ can accommodate both $\cC_{j}$ and $T$; then we move both 
$T$ and all documents from $\cC_j$ into $\cC_{j+1}$\footnote{Please note the difference between Transformations~\ref{trans:trans1} and~\ref{trans:trans1worst}. In Transformation~\ref{trans:trans1} we look for the sub-collection $\cC_j$ that can accommodate the new document and \emph{all} smaller sub-collections $\cC_0$, $\ldots$, $\cC_{j-1}$. In Transformation~\ref{trans:trans1worst} we look for the sub-collection $\cC_{j+1}$ that can accommodate that can accommodate the new document and the preceding sub-collection $\cC_j$. We made this change in order to avoid some technical complications caused by delayed re-building.}. If the new document $T$ is large, $|T|\ge \max_j/2$, we can afford to re-build $\cC_{j+1}$ immediately after the insertion of~$T$. 
If the size of~$T$ is smaller than $\max_j/2$, re-building of~$\cC_{j+1}$ is postponed. For every following update, we spend $O(\log^{\eps}n\cdot u(n))$ time per symbol on creating the new version of~$\cC_{j+1}$. The old versions of~$\cC_j$, $\cC_{j+1}$ are retained until the new version is completed. 
If the number of symbols that are marked as deleted in $\cC_j$ exceeds $\max_j/2$, we employ the same procedure for moving $\cC_j$ to $\cC_{j+1}$: $\cC_j$ is locked and  we start the process of constructing a new version $\cC_{j+1}$ that contains all undeleted documents from $\cC_j$.

The disadvantage of delayed re-building is that we must keep two copies of every document in $\cC_j\cup \cC_{j+1}$ until new $\cC_{j+1}$ is completed. In order to reduce the space usage, we keep only a fraction of all documents in sub-collections $\cC_i$. All $\cC_i$ for $0\le i\le r$ will contain  $O(n/\tau)$ symbols, where $\tau$ is the parameter determining the trade-off between space overhead and query time. The remaining documents are kept in top sub-collections $\cT_1$, $\ldots$, $\cT_g$ where $g\le 2\tau$. 
Top sub-collections are constructed using the same delayed  approach. But once $\cT_i$ is finished, no new documents are inserted 
into $\cT_i$. We may have to re-build a top collection or merge it with another $\cT_j$ when the fraction  of deleted symbols exceeds a threshold value $1/\tau$.  We employ the same rebuilding-in-the-background approach. However, we will show that the background procedures for maintaining $\cT_i$ can be scheduled in such a way that only one $\cT_j$ is re-built at any given moment. Hence, the total  space 
overhead due to re-building and storage of deleted elements is bounded by an additive term $O(n(\log \sigma+ w(n))/\tau)$.

\paragraph{Data Structures}
We split  a document collection $\cC$ into subcollections $\cC_0$, $\cC_1$, $\ldots$, $\cC_r$, $\cL_1$, $\ldots$, $\cL_r$ and top subcollections $\cT_1$, $\ldots$, $\cT_g$ where $ g=O(\tau)$. We will also use auxiliary collections $\cN_1$, $\ldots$, $\cN_{r+1}$ and temporary collections $Temp_1$, $\ldots$, $Temp_r$. $Temp_i$ are also used to answer queries but each non-empty $Temp_i$ contains exactly one document; $Temp_i$ are used as temporary storage 
for new document that are not yet inserted into ``big'' collections. 
The sizes of sub-collections can be defined as a function of parameter $n_f$ such that $n_f=\Theta(n)$; the value 
of~$n_f$ changes when $n$ becomes too large or too small. Let $\max_i=2(n_f/\log^2n)\log^{i\eps}n$.
We maintain the invariant $|\cC_i|\le \max_i$ for all $i$, $0\le i \le r$,  but $r$ is chosen in such way that $n_f/\log^{2-r\eps}n_f=n_f/\tau$. 
Every $\cT_i$ contains  $\Omega(n_f/\tau)$ symbols.  If $\cT_i$ contains 
more than one text, then its size is at most $4n_f/\tau$; otherwise $\cT_i$ can be arbitrarily large. When a collection $\cC_j$ is merged with $\cC_{j+1}$, the process
of re-building $\cC_j$ can be distributed among a number of future updates (insertions and deletions of documents). During this time $\cC_j$ is locked: we set $\cL_j=\cC_j$ and 
initialize a new empty sub-collection $\cC_j$. When a new subcollection $\cN_{j+1}=\cC_{j+1}\cup\cC_j$ is completed, we set $\cC_{j+1}=\cN_{j+1}$ and discard old $\cC_{j+1}$ and $\cL_j$. A query is answered by 
querying all non-empty $\cC_i$, $\cL_i$, $Temp_i$, and $\cT_i$. Therefore the cost of answering a range-finding query grows by $O(\tau)$. The costs of locating and extracting are the same as in the static index. We show main  sub-collections used by our method on Figure~\ref{fig:cascad-worst}. 
\begin{figure*}[tb]
  \centering
  \includegraphics[width=.6\textwidth]{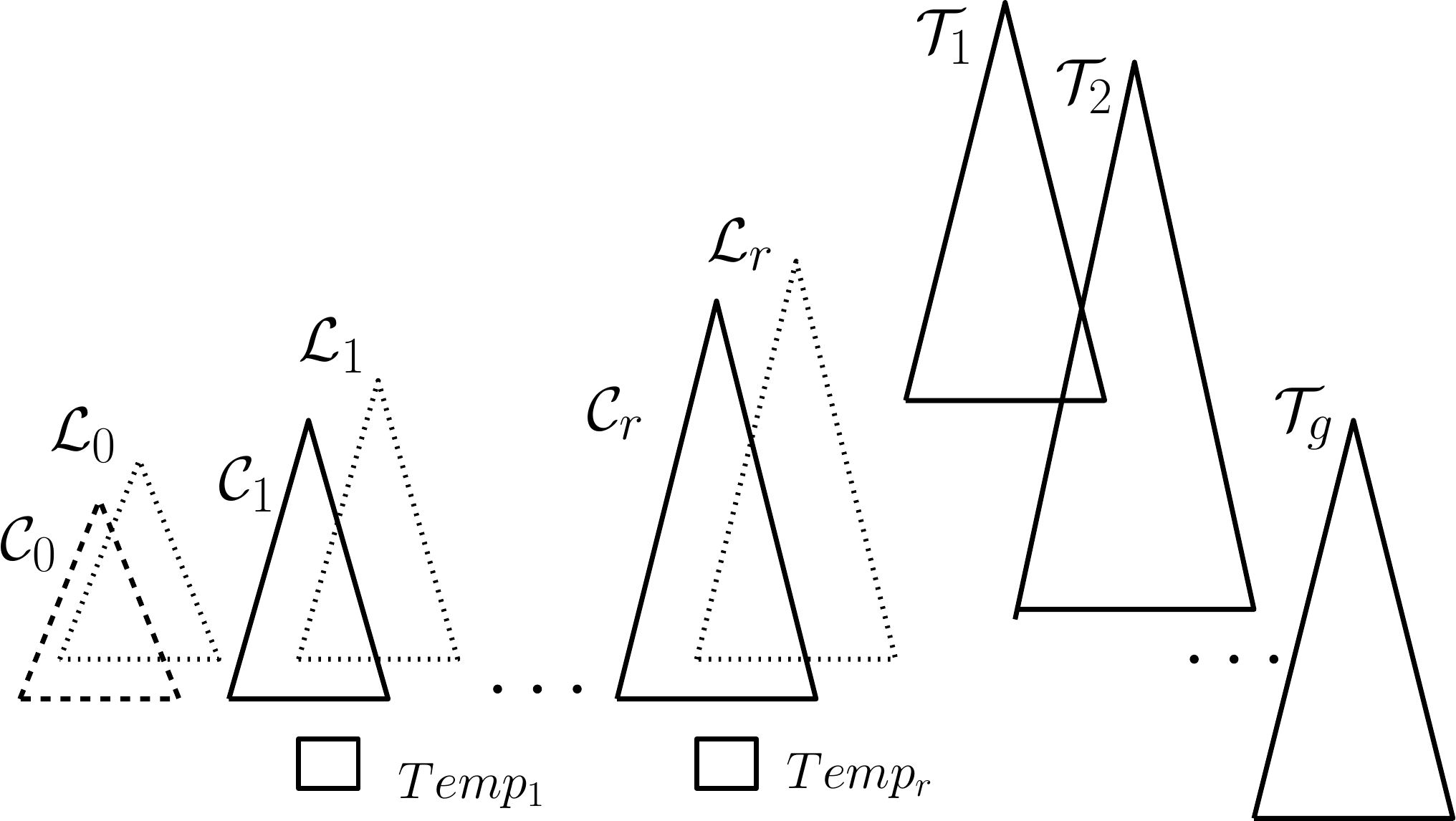}
  \caption{Dynamization with worst-case update guarantees. Only main sub-collections used for answering queries are shown. $\cL'_r$ and auxiliary 
collections $\cN_i$ are not shown. }
  \label{fig:cascad-worst}
\end{figure*}
%[height=.2\textheight]

\paragraph{Insertions}
When a document $T$ is inserted, we consider all $j$, such that $0\le j \le r$ and the data structure $\cL_j$ is not empty. For every such $j$, we  spend $O(|T|\log^{\eps}n\cdot u(n))$ units of time on constructing $\cN_{j+1}$. If $\cN_{j+1}$ for some $0\le j \le r-1$ is completed, we set $\cC_{j+1}=\cN_{j+1}$ and $\cN_{j+1}=\Temp_{j+1}=\emptyset$; if $\cN_{r+1}$ is completed, we set $\cT_{g+1}=\cN_{r+1}$, increment the number of top collections $g$, and set  $\cN_{r+1}=\Temp_{r+1}=\emptyset$. 
%If $\cL_r$ exists, we also spend $O(|T|\log^{\eps}n\cdot u(n))$ time constructing a new data structure $\cT_i$. 
Then we look for a sub-collection that can accommodate the new document $T$. If $|T|\ge n/\tau$, we create the index for a new sub-collection $\cT_i$ that contains a single document  $T$. 
If $|T|< n/\tau$, we look for  the smallest index $j$, such that $|\cC_{j+1}|+|\cC_j|+ |T|\le \max_{j+1}$. That is, $\cC_{j+1}$ can accommodate both the preceding sub-collection $\cC_j$ and $T$. If $|T|\ge \max_{j}/2$, we set $\cC_{j+1}=\cC_j\cup \cC_{j+1}\cup T$ 
and create an index for the new $\cC_{j+1}$ in $O(|\cC_{j+1}|\cdot u(n))=O(|T|\log^{\eps}n\cdot u(n))$ time. If $|T|<\max_j/2$, the collection $\cC_j$ is locked. We set 
$\cL_j=\cC_j$, $\cC_j=\emptyset$ and initiate the process of creating $\cN_{j+1}=\cC_j\cup \cC_{j+1}\cup T$. The cost of creating the new index for $\cN_{j+1}$ will be distributed among the next $\max_j$ update operations.  We also create a temporary static index $\Temp_{j+1}$  for  the text $T$ in $O(|T|u(n))$ time. This procedure is illustrated on Fig.~\ref{fig:insertworstcase}.
If the index $j$ is not found and 
$|\cC_i|+|\cC_{i+1}|+|T|> \max_{i+1}$ for all $i$, $0\le i < r$, we lock $\cC_r$ (that is, set $\cL_r=\cC_r$ and $\cC_r=\emptyset$) and initiate the process of constructing $\cN_{r+1}=\cL_r\cup T$. We also create a temporary index $Temp_{r+1}$ for the document $T$ in 
$O(|T|u(n))$ time.
\begin{figure*}[tb]
\centering
\begin{tabular}{ccccc}
  \includegraphics[width = .25\textwidth]{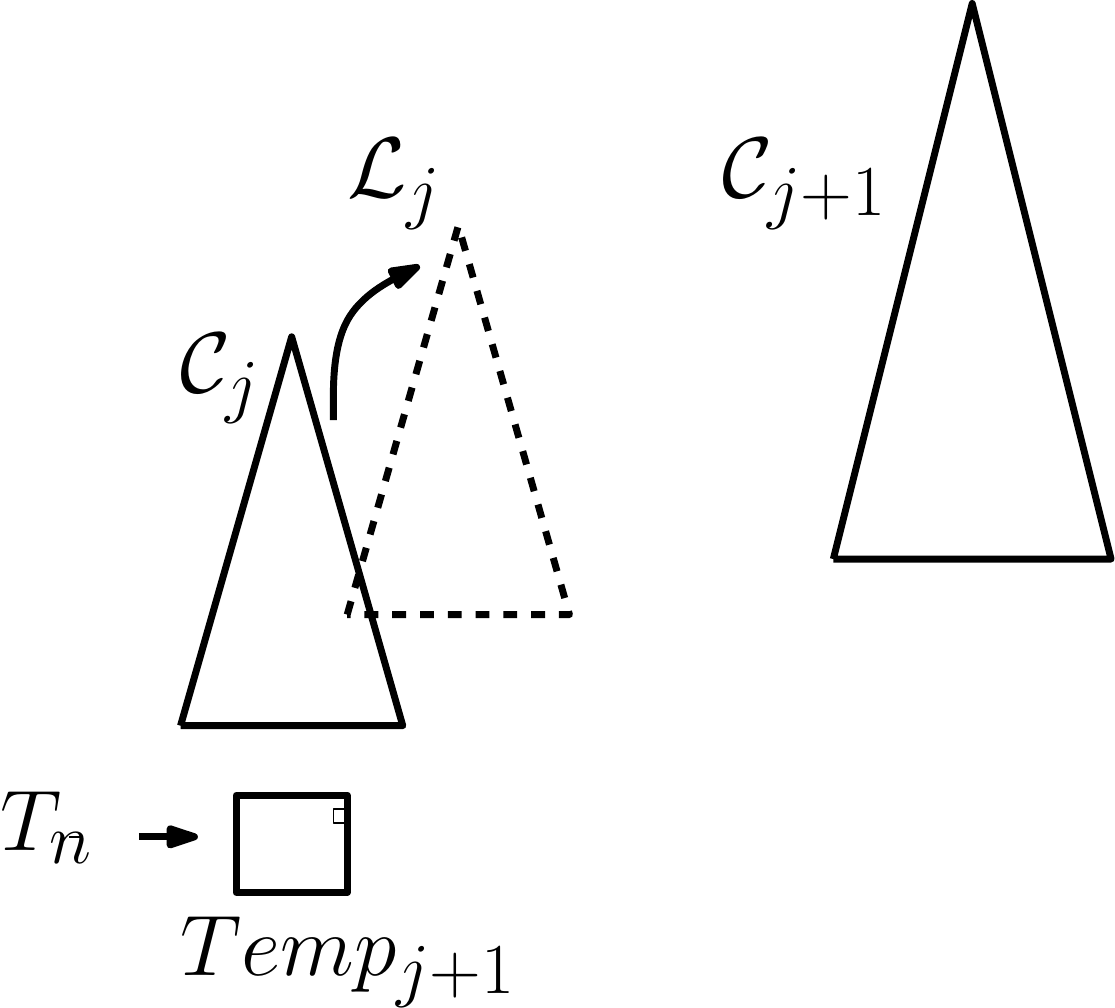}
 & \hspace*{.15cm} &
\includegraphics[width = .25\textwidth]{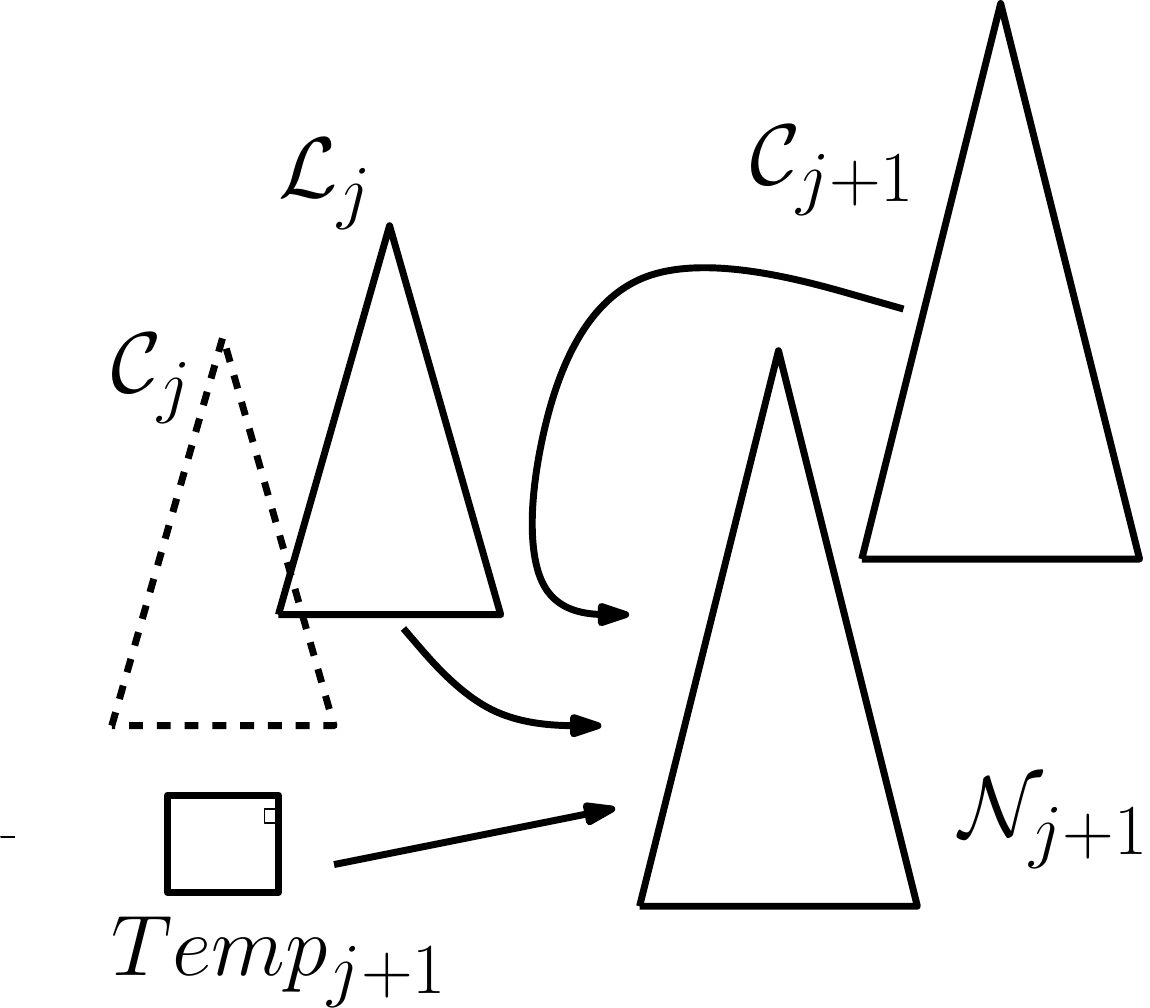} &\hspace*{.15cm} &
\includegraphics[width = .25\textwidth]{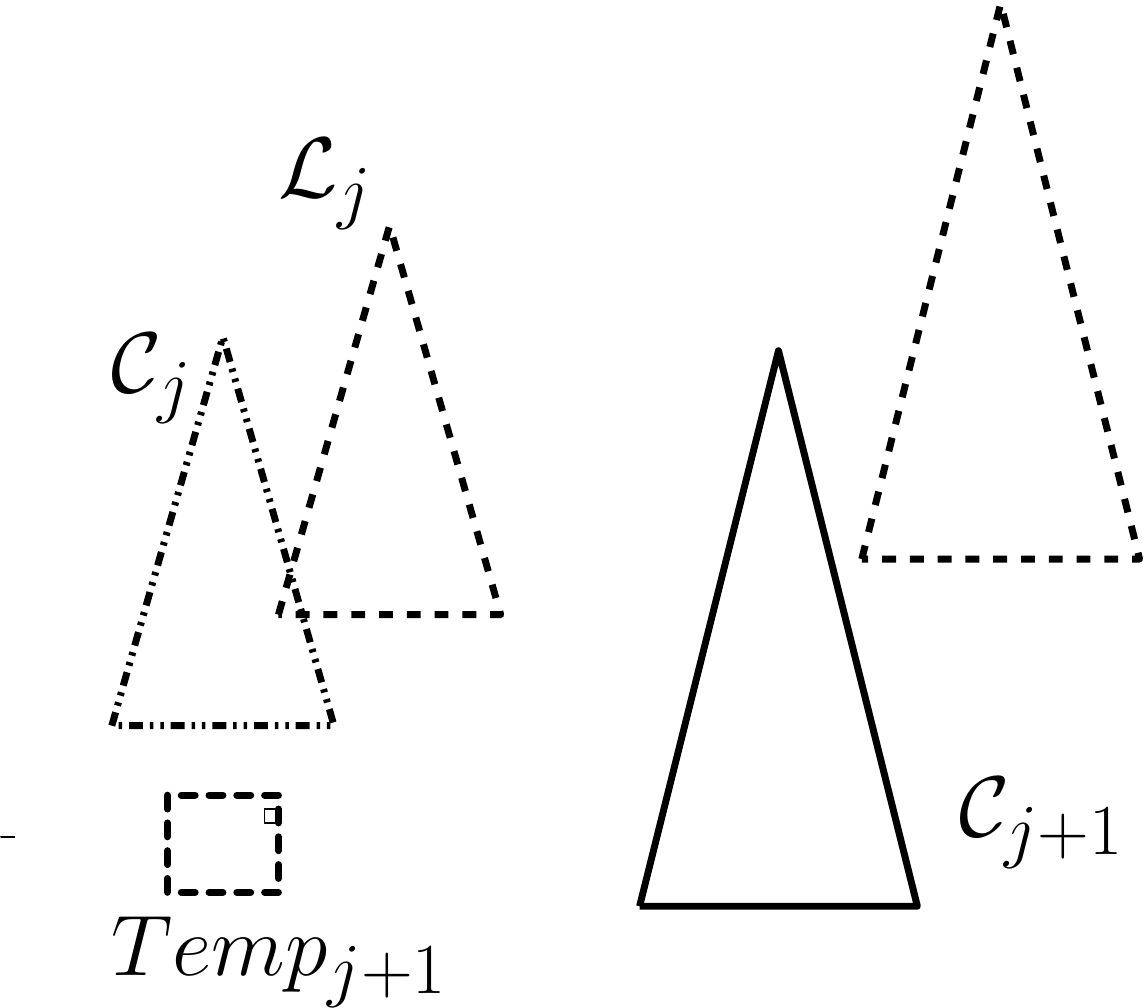} \\
 {\bf (a)}   &  & {\bf(b)} &  & {\bf (c)} \\
\end{tabular}
% \begin{tabular}{ccc}
% \includegraphics[width = .30\textwidth]{update-rebuild}
%  & \hspace*{.5cm} &
% \includegraphics[width = .30\textwidth]{update-rebuild2}\\
% {\bf (a)}   &  & {\bf(b)}  \\
% \includegraphics[width = .30\textwidth]{update-rebuild3} & & \\
% {\bf (c)} & & \\
% \end{tabular}
\caption{Suppose that $\cC_{j+1}$ is the first sub-collection that can accommodate both $\cC_j$ and a new document $T_n$. If $\cC_j$ must be  rebuilt in the background, we ``rename'' $\cC_j$ to $\cL_j$ and initialize another (initially empty) $\cC_j$.
  New document $T_n$ is put into a separate collection $Temp_{j+1}$ ({\bf a}). A background process creates a new collection $\cN_{j+1}$ that contains all documents from $\cL_j$, $\cC_{j+1}$ and $Temp_{j+1}$ ({\bf b}). When $\cN_{j+1}$ is finished, we discard $\cC_{j+1}$, $\cL_j$ and $Temp_{j+1}$, and set $\cC_{j+1}=\cN_{j+1}$ ({\bf c}). Our procedure guarantees that $\cN_{j+1}$ is completed before the new sub-collection $\cC_j$ must be re-built again.
} 
\label{fig:insertworstcase}
\end{figure*}

% We must also take care that the number of top sub-collections is not too high. If the number of top sub-collections equals to $2\tau$, we start merging the top collections. We set $n_f$ to be equal to the current total size of all sub-collections. 
% We also increase the value of parameter $r$, so that $\max_r= n_f/\tau$.
% All top sub-collections $\cT_i$ that contain less than $n_f/\tau$ symbols are merged into new collections $\cT'_l$ of total size between $n_f/\tau$ and $2n_f/\tau$ symbols. During the next $n_f/\tau$ symbol updates (that is, insertions and deletions of texts of total size $n_f/\tau$), we construct new collections $\cT'_l$. 

\paragraph{Deletions}
Indexes for sub-collections $\cC_i$, $1\le i \le r$,  and $\cT_j$, $1\le j \le g$, support lazy deletions in the same way as in Section~\ref{sec:dyncoll}: when a document is deleted from a sub-collection, we simply mark the positions of suffixes in the suffix array as deleted and set the corresponding bits in the bit vector $B$ to $0$. Augmenting an index so that lazy deletions are supported is done in exactly the same way as in Section~\ref{sec:dyncoll}. 

We will need one additional sub-collection $\cL'_r$ to support deletions. 
If a sub-collection $\cC_j$ for $1\le j \le r-1$ contains $\max_j/2$ deleted elements, we start the process of re-building $\cC_j$ and merging it with $C_{j+1}$. This procedure is the same as in the case of insertions. We lock $\cC_j$ 
by setting $\cL_j=\cC_j$ and $\cC_j=\emptyset$. The data structure $\cN_{j+1}=\cC_{j+1}\cup \cL_j$ will be re-built during the following $\max_j/2$ updates. \todo{check the last sentence!}
If a sub-collection $\cC_r$ contains $\max_r/2$ deleted symbols, we set $\cL'_r=\cC_r$ and $\cC_r=\emptyset$. The sub-collection $\cL'_r$ will be merged with the next sub-collection $\cT_i$ to be re-built. 

If a collection $\cT_i$ contains a single document and this document is deleted, then $\cT_i$ is discarded.
We also bound the number of deleted symbols in any $\cT_i$ by $n_f/\tau$. This is achieved by running the following background process. 
After each series of~$n_f/(2\tau\log \tau)$ symbol deletions, we identify $\cT_j$ that contains the largest number of deleted symbols. 
During the next $n_f/(2\tau\log\tau)$ symbol deletions we build the new index for $\cT_j$ without the deleted symbols.
At the same time  we remove the deleted symbols from $\cL'_r$ if $\cL'_r$ exists. 
If $\cL'_r$ exists and contains at least $n_f/2\tau$ undeleted symbols, we create an index for a new sub-collection $\cT'_{g+1}$ and increment the number $g$  of top collections.  If $\cL'_r$ exists, but contains less than
$n_f/2$ undeleted symbols, we merge $\cL'_r$ with the largest $\cT_j$ that contains more than one document and split the result if necessary: if the number of undeleted symbols in $\cL'_r\cup T_j$ does not exceed $2n_f/\tau$, we construct an index for $T_j\cup \cL'_r$ without deleted symbols; otherwise, we split $T_j\cup \cL'_r$ into two parts $T^1_j$, $T^2_j$ and create indexes for 
the new sub-collections. Our method  guarantees us that the number of deleted elements in any collection $\cT_i$  does 
not exceed $O(n_f/\tau)$ as follows from a Theorem of Dietz and Sleator~\cite{DietzS87}.
\begin{lemma}[\cite{DietzS87}, Theorem 5]
  \label{lemma:dietz}
Suppose that $x_1$,$\ldots$, $x_g$ are variables that are initially zero. Suppose that the following two steps are iterated: 
(i) we add a non-negative real value $a_i$ to each $x_i$ such that $\sum a_i=1$ (ii) set the largest $x_i$ to $0$. 
Then at any time $x_i\le 1+ h_{g-1}$ for all $i$, $1\le i\le g$, where $h_i$ denotes the $i$-th harmonic number. 
\end{lemma}
Let $m_i$ be the number of deleted elements in the $i$-th top collection $\cT_i$ and $\delta=n_f/(2\tau \log \tau)$. We define $x_i=m_i/\delta$.
We consider the working of our algorithm during the period when the value of~$n_f$ is fixed. Hence, $\delta$ is also fixed
 and the number of variables $x_i$ is $O(\tau)$ (some $x_i$ can correspond to empty collections). 
 Every iteration of the background process sets the largest $x_i$ to $0$. During each iteration $\sum x_i$ increases by 1. Hence, the values of~$x_i$ can be bounded from above by the result of Lemma~\ref{lemma:dietz}: $x_i\le 1+h_{2\tau}$  for all $i$ at all times.
Hence $m_i=O((n_f/2\tau\log\tau)\log\tau)=O(n_f/\tau)$ for all $i$ because $h_i=O(\log i)$. 
Thus the fraction of deleted symbols in each 
$\cT_i$ is $O(1/\tau)$. 

It is easy to show that the sub-collections that we use are sufficient for our algorithm. When a sub-collection $\cL_j$ 
is initialized, $\cC_j$ is empty. The situation when $\cC_j$ cannot accommodate a new document $T_n$ and a preceding  subcollection $\cC_{j-1}$ can happen  only after $\max_j-\sum_{t=1}^{j-1}\max_{t}$ new symbol insertions. Since we spend $O(\log^{\eps}n\cdot u(n))$ time for constructing $\cN_{j+1}$ with each new symbol insertion, we can choose constants in such a way
that construction of~$\cN_{j+1}$ is finished (and $\cL_j$ is discarded) after $\max_j/2<\max_j-\sum_{t=1}^{j-1}\max_{t}$ symbol insertions. The situation when $\cC_j$ contains 
$\max_j/2$ deleted symbols can happen after at least $\max_j$ new symbol updates ($\max_j/2$ insertions and $\max_j/2$ deletions). Hence, the collection $\cL_j$ is discarded before $\cC_j$ has to be locked again. In our description of update procedures we assumed that the parameter $n_f$ is fixed. We can maintain the invariant $n_f=\Theta(n)$ using standard methods; for completeness \shlongver{we will provide a description in the full version.}{we provide a description in Section~\ref{sec:updatesapp}. }

The space overhead caused by storing copies of deleted elements is bounded by 
$O(n/\tau)$: all $\cC_i$ contain $O(n/\tau)$ symbols and at most every second symbol in each $\cC_i$ is from a deleted document; the fraction of deleted symbols in each $\cT_i$ does not exceed $O(1/\tau)$. By the same argument,
at any moment of time at most $O(n/\tau)$ symbols are in sub-collections that are re-built. Hence re-building procedures running in the background need $O(nw(n)/\tau)$ bits of space. Since each $\cT_i$ contains at most $O(|\cT_i|/\tau)$ deleted symbols,
we can store the data structure $V$, which enables us to identify undeleted elements in any range of the suffix array and is implemented as described in Lemma~\ref{lemma:fastbit2}, using $O(|\cT_i|\log\tau/\tau)$ bits. Data structures $V$ for all $\cT_i$ need $O(n\log\tau/\tau)$ bits.
Hence, the total space overhead of~$\cI_d$ compared to $\cI_s$ is 
$O(n\frac{w(n)+\log\tau+\log\sigma}{\tau})$ bits.

\paragraph{Counting Occurrences}
Our dynamic indexes can be easily extended so that pattern counting queries are supported. 
\begin{theorem}
We can augment the indexes $\cI_d$ of Transfomations~\ref{trans:trans1}-~\ref{trans:trans1worst} with $O((n\log\tau)/\tau)$ additional bits so that all occurrences of a pattern can be counted in $O(\tcount)$ time, where $\tcount=(\trange+\log n/\log\log n)(r+\tau)$ and $\tau$ is defined as in the proofs of respective Transformations. If counting is supported, update times are increased by $O(\log n/\log\log n)$ additive term per symbol.
\end{theorem}
\begin{proof}
Every semi-dynamic index for a sub-collection $\cC_i$ (respectively $\cT_i$)
already keeps a vector $B$ that enables us to identify the suffixes of already deleted documents in the suffix array. We also store each $B$ in a data structure of 
Navarro and Sadakane~\cite{NS10} that supports rank queries in $O(\log n/\log\log n)$ time and updates in $O(\log n/\log\log n)$ time. If $B$ contains $O(|B|/\tau)$ zero values, then the structure of~\cite{NS10}  needs $O((|B|/\tau)\log\tau)$ bits. Using this data 
structure, we can count the number of~$1$'s in any portion $B[a..b]$ of~$B$ in the same time.  To answer a counting query, we first answer a range-finding query in every sub-collection. For every non-empty range that we found, we count the number of~$1$'s in that range. Finally, we sum the answers for all sub-collections.
Since a range-finding query returns the range of all suffixes  that start with a query pattern and each $1$ in $V$ corresponds to  a suffix of  an undeleted document, our procedure is correct.
\end{proof}

\section{Dynamic Indexes}
\label{sec:dynind}
To obtain our results on dynamic document collections, we only need to plug some currently known static indexes into Transformations~\ref{trans:trans1} and~\ref{trans:trans1worst}.
\shlongver{We will prove the statements about constructibility of static indexes in the full version of this paper.}{For completeness, we prove the statements about constructibility of static indexes in Section \ref{sec:constimeapp}.}
% The static index of  Ferragina et al.~\cite{FMMN07}  for alphabets of poly-logarithmic size is $(\log n/\log \log n, o(1))$-constructible. For completeness, we prove the statements about constructibility of static indexes in Section \ref{sec:constimeapp}.
% This index  achieves $\trange=O(|P|)$, $\textract=O(s+\ell)$, and $\tsa=\tlocate=O(s)$; it uses $nH_k+O(n\frac{\log n}{s})+o(n\log\sigma)$ bits.
% We apply Transformation~\ref{trans:trans1} with $\tau=\log \log n$ to this index. We obtain a dynamic index with  $\trange=O(|P|)$, $\textract=O(s+\ell)$,  $\tlocate=O(s)$, 
% and $\tcount=O(|P|+\log n/\log \log n)$; it  uses $nH_k+O(n\frac{\log n}{s})+o(n\log\sigma)+O(n(\log\log\log n/\log\log n))= nH_k+O(n\frac{\log n}{s})+o(n\log\sigma)$ bits. 
% Insertions and deletions are supported in $O(|T|\log^{1+\eps}n)$  and $O(|T|(\log^{1+\eps}n+s))$ amortized time. We can assume that $s=O(\log n \log \log n)$ because $s=\log n\log\log n$ is sufficient to obtain asymptotic space optimality. Hence, the deletion time can be simplified to $O(|T|\log^{1+\eps}n)$. 
 
The static index of Belazzougui and Navarro~\cite{BelazzouguiN11}  is $(\log^{\eps}n,\log\sigma)$-constructible.  Their index  achieves $\trange=O(|P|)$, $\textract=O(s+\ell)$, and $\tsa=\tlocate=O(s)$ for arbitrarily large alphabets; it needs $nH_k + O(n\frac{\log n}{s}) + O(n\frac{H_k}{\log\log n})+O(n)$ bits. We apply Transformation~\ref{trans:trans1worst} with $\tau=\log\log n$. The construction algorithm for this index relies on randomized algorithm for constructing an mmphf functions~\cite{BelazzouguiN11}; therefore the update procedures of our dynamic data structure also rely on randomization in this case. The resulting dynamic index uses $nH_k+O(n\frac{\log n}{s})+ O(n\frac{\log\sigma}{\log\log n})+O(n)$ bits. This index achieves 
$\trange=O(|P|\log\log n)$, $\textract=O(s+\ell)$,  $\tlocate=O(s)$. Insertions and deletions are supported in $O(|T|\log^{\eps}n)$ time and $O(|T|(\log^{\eps}n+s))$ expected time respectively. If counting queries are also supported, then $\tcount= O(|P|\log \log n +\log n)$ and updates take $O(|T|\log n)$ expected time. 

The index of Barbay et al.~\cite{BGNN10,BarbayCGNN14} is also $(\log^{\eps}n, \log\sigma)$-constructible and uses  $nH_k+O(n\frac{\log n}{s})+o(n\log\sigma)$ bits.  If the alphabet size $\sigma=\log^{O(1)}n$, this index  achieves $\trange=O(|P|)$, $\textract=O(s+\ell)$, and $\tlocate=O(s)$; it uses $nH_k+O(n\frac{\log n}{s})+o(n\log\sigma)$ bits. If we set $\tau=\log\log n$ and apply Transformation~\ref{trans:trans1worst}, we obtain a dynamic data structure with $\trange=O(|P|\log\log n)$, $\textract=O(s+\ell)$, and $\tsa=\tlocate=O(s)$. For an arbitrary alphabet size $\sigma$, the index of Barbay et al.~\cite{BGNN10,BarbayCGNN14} achieves $\trange=O(|P|\log\log \sigma)$, $\textract=O((s+\ell)\log\log \sigma)$, and $\tsa=\tlocate=O(s\log\log \sigma)$. Again we set $\tau=\log\log n$ and apply Transformation~\ref{trans:trans1worst}. We obtain a dynamic index that has query costs  $\trange=O(|P|\log\log \sigma\log\log n)$, $\textract=O((s+\ell)\log\log \sigma)$, and $\tsa=\tlocate=O(s\log\log \sigma)$. 
Insertions and deletions are supported in $O(|T|\log^{\eps}n)$ time and $O(|T|(\log^{\eps}n+s))$ time respectively.
If counting queries are also supported, then $\tcount= O(|P|\log \log n \log\log\sigma +\log n)$
(resp.\ $\tcount= O(|P|\log \log n +\log n)$ if $\sigma=\log^{O(1)}n$) and updates take $O(|T|\log n)$ time. 

The index of Grossi and Vitter~\cite{GrossiV05} is $(\log^{\eps}n,\log\sigma)$-constructible.  It achieves  $\tlocate=O(\log^{\eps}n)$, $\trange=O(|P|/\log_{\sigma}n +\log^{\eps}n)$ and $\textract=O(\ell/\log_{\sigma}n)$. 
We apply Transformation~\ref{trans:trans1worst} with $\tau=1/\delta$ for a constant $\delta$. The resulting dynamic index uses $O(n\log\sigma(1+1/\delta))=O(n\log\sigma)$ bits and has the following query costs: $\tlocate=O(\log^{\eps}n)$, $\trange=O(|P|/\log_{\sigma}n +\log^{\eps}n)$, $\textract=O(\ell/\log_{\sigma}n)$.\no{, and $\tcount=O(|P|/\log_{\sigma}n+\log n/\log\log n)$.} 
\shlongver{As will be shown in the full version,}{ 
As described in Section~\ref{sec:uncompr},} in this case the data structure for uncompressed sequence $\cC_0$ relies on hashing. Therefore the update procedure is randomized. Updates are supported in $O(|T|\log^{2\eps}n)$ expected time, but we can replace $\eps$ with $\eps/2$ in our construction and reduce the update time to $O(|T|\log^{\eps}n)$. If counting queries are also supported, then $\tcount=O(|P|/\log_{\sigma}n+\log n/\log\log n)$ and updates take $O(|T|\log n)$ expected time. 
If we want to support updates using a deterministic procedure, then the cost of searching in $\cC_0$ grows to 
$O(|P|(\log\log n)^2/\log_{\sigma}n+ \log n)$. In this case $\trange=\tcount=O(|P|(\log\log n)^2/\log_{\sigma}n +\log n)$,
$\tlocate=O(\log^{\eps}n)$, and  $\textract=O(\ell/\log_{\sigma}n)$.

\section{Dynamic Graphs and Binary Relations}
\label{sec:dyngraph}
Let $R$  denote a binary relation between $t$ objects and $\sigma_l$ labels. In this section we denote by $n$ the cardinality of~$R$, i.e., the number of object-label pairs. We will assume that objects and labels are integers from intervals $[1,\sigma_l]$ and $[1,t]$ respectively. Barbay et al.~\cite{BGMR07} showed how a static relation $R$ can be represented by a string $S$. % Let $M$ denote a conceptual binary matrix that represents $R$; columns of~$M$ correspond to objects and rows of~$M$ correspond to labels. Barbay et al.~\cite{BGMR07} store $M$ as a binary string $B$ of length $\sigma n$ with $t$ $1$-bits and $n/\sigma$ strings $S_i$ over an alphabet $\sigma$.  They showed how basic queries on $R$ can be answered by rank, access, and select  queries  on $B$ and $S_i$. 
A dynamization of their  approach based on dynamic data structures for rank and select queries is described in~\cite{NavarroN13}. 

Let $M$ be a matrix that represents a binary relation $R$; columns of~$R$ correspond to objects and rows correspond to matrices. The string $S$ is obtained by traversing $M$ columnwise (i.e., objectwise) and writing the labels. An additional bit string $N$ encodes the numbers of labels related to objects: $N=1^{n_1}01^{n_2}0\ldots 1^{n_t}$, where $n_i$ is the number of labels related to the $i$-th object. Using rank, select, and access queries on $N$ and $S$, we can enumerate objects related to a label, enumerate labels related to an object, and decide whether an object and a label are related. 

\paragraph{Deletion-Only Data Structure}
We keep $R$ in $S$ and $N$ described above; $S$ and $N$ are stored in static data structures.  If a pair $(e,l)$  is deleted from $R$, we find the element of~$S$ that encodes this pair and mark it as deleted.  
We record marked elements (i.e. pairs that are deleted but are still stored in the data structure) in a bit vector $D$: $D[i]=0$ if and only if the pair $S[i]$ is marked as deleted.  We maintain the data structure of Lemma~\ref{lemma:fastbit2} on $D$. Moreover we keep $D$ in a data structure described in~\cite{GonzalezN09}; this data structure enables us to count 
the number of~$1$-bits in any range of~$D$.
For each label $a$ we also keep a data structure $D_a$. $D_a$ is obtained by traversing the $a$-th row of~$M$:  if  $M[a,j]\not=0$, then we append  $0$ to $D_a$ if $(a,j)$ is marked as deleted; if $M[a,j]\not=0$ and $(a,j)$ is not marked as deleted, we append $1$ to $D_a$.  For each $D_a$ we also maintain data structures for reporting and counting $1$-bits described above.  Finally we record indices of deleted labels and objects in two further bit sequences. The static data structures on $S$ and $N$ are implemented as in\cite{BarbayCGNN14}, so that rank and select queries are answered in $O(\log\log \sigma_l)$ time and any $S[i]$ or $N[i]$ can be retrieved in constant time. 

If we need to list labels related to an object $i$, we first find the part of~$S$ that contains these labels. Let $l=\ra_1(\sel_0(i-1,N),N)$ and $r=\ra_1(\sel_0(i,N),N)$. We list all elements of~$S[l..r]$ that are not marked as deleted by enumerating all $1$-bits in $D[l..r]$. Then we access and report $S[i_1]$, $S[i_2]$, $\ldots$, $S[i_f]$, where $i_1$, $i_2$, $\ldots$, $i_f$ are positions of~$1$-bits in $D[l..r]$. 
In order to list objects related to a label $a$, we find positions of~$1$-bits in $D_a$. Then we access and report $\sel_a(j_1,S)$, $\sel_a(j_2,S)$, $\ldots$, where 
$j_1$, $j_2$, $\ldots$ denote positions of~$1$-bits in $D_a$. In order to determine whether an object $i$ and a label $a$ are related, we compute $d=\ra_a(r,S)-\ra_a(l,S)$, where $l$ and $r$ are as defined above. If $d=0$, then the object $i$ and the label $a$ are not related. If $d=1$, we compute $j=\sel_a(\ra_a(r,S),S)$; $i$ and $a$ are related if and only if 
$D[j]=1$.

When $(e,l)$ is deleted, we  find the position $j$ of~$(e,l)$ in $S$ and set $D[j]=0$; $j$ can be found with a constant number of rank and select queries. We also set $D_a[j']=0$ for $j'=\ra_a(S,j)$. When  an empty label or an empty object is removed, we simply record this fact by adding it to a compact list of empty labels (resp.\ empty objects).  When the number of pairs that are marked as deleted exceeds $n/\tau$, we start the process of re-building the data structure. The cost of re-building is distributed among the following updates; we will give a more detailed description in the exposition of the fully-dynamic data structure.

\paragraph{Fully-Dynamic Data Structure} 
We split a binary relation $R$, regarded as a set of object-label pairs, into subsets and keep these subsets in data structures $\bC_0$, $\bC_1$, $\ldots$, $\bC_r$, $\bL_1$, $\ldots$, $\bL_r$, and $\bT_1$, $\ldots$, $\bT_{g}$ for $g=\Theta(\tau)$.  We set the parameter $\tau=\log\log n$. Only $\bC_0$ is stored in a fully-dynamic data structure, but we can afford to keep $\bC_0$ in $O(\log n)$ bits per item because it contains only a small fraction of pairs.
All other pairs are stored in deletion-only data structures described above. Distribution of pairs among subsets and procedures for re-building deletion-only data structures are the same as in Section~\ref{sec:worst}. To simplify a description, we will not distinguish between a subset and a data structure that stores it. 

$\bC_0$ contains at most $\max_0=2n/\log^2 n$ pairs. Each structure $\bC_i$ for $r\ge i\ge 1$ contains at most $\max_i=2n/\log^{2-i\eps}n$ pairs. Every $\bT_i$ contains 
at most $2n/\tau$ pairs. 
Data structure $\bC_0$ contains object-label pairs in uncompressed form and uses $O(\log n)$ bits per pair. For every object $i$ 
 that occurs in $\bC_0$ we keep a list $L_i$ that contains all labels that occur in pairs $(i,\cdot)\in \bC_0$; for each label $a$ that occurs in $\bC_0$ we keep a list of objects that occur in pairs $(\cdot,a)\in \bC_0$. Using these lists we can enumerate all objects related to a label or labels related to an object in $\bC_0$ in $O(1)$ time per datum. If we augment lists $L_i$ with predecessor data structures described in~\cite{AnderssonT07}, we can also find out whether an object $i$ and a label $a$ are related in $O((\log\log\sigma_l)^2)$ time. 

All pairs in $\bC_1$,$\bL_1$,$\ldots$, $\bC_r$, $\bL_r$, and $\bT_1$, $\bT_{\tau}$ are kept in deletion-only data structures described above. A new object-label pair  $(i,a)$ is inserted into $\bC_0$ if $\bC_0$ contains less than $\max_0$ pairs. Otherwise we look for the smallest $j$, $0\le j< r$, such that $|\bC_{j+1}|+|\bC_j|+1\le \max_{j+1}$. We lock $\bC_j$ by setting $\bL_j=\bC_j$, $\bC_j=\emptyset$ and  initiate the process of creating $\bN_{j+1}=\bC_j\cup\bC_{j+1}\cup\{(i,a)\}$. If $|\bC_{i+1}|+|\bC_i|+1\le \max_{i+1}$ for all $i<r$, we lock $\bC_r$ and start the process of constructing $\bN_{j+1}=\bC_r\cup \{(i,a)\}$. The cost of creating $\bN_j$ is distributed among the next $\max_{j}$  updates in the same way as in Section~\ref{sec:worst}. We observe that data structures $Temp_i$ are not needed now because each update inserts only one element (pair) into the relation $R$.  We guarantee that each structure $\bC_i$ for some $1\le i\le r$ contains at most $\max_i/2$ pairs marked as deleted and  $\bT_i$ for $1\le i\le r$ 
contains an $O(1/\tau)$ fraction of deleted pairs. Procedures for re-building data structures that contain too many pairs marked as deleted are the same as in Section~\ref{sec:worst}.

Our fully-dynamic data structure must support insertions and deletions of new objects and labels. An object that is not related to any label or a label that is not related to any object can be removed from a data structure. This means that both the number of labels $\sigma_l$ and the number of objects $t$ can change dynamically. Removing and inserting labels implies changing the alphabets 
of strings $S$ that are used in deletion-only data structures. 
Following~\cite{NavarroN13} we store two global tables, $NS$ and $SN$; $SN$  maps labels to integers bounded by $O(\sigma_l)$ (global label alphabet) and $NS$ maps integers back to labels. We also keep bitmaps $GC_i$ and $GT_i$, $GL_i$, and $GN_i$ for all subsets $C_i$, $L_i$, $N_i$, and $T_i$. $GC_i[j]=1$ if the label that is assigned to integer $j$ occurs in $\bC_i$ and $GC_i[j]=0$ otherwise; $GT_i$, $GL_i$, and $GN_i$ keep the same information for subsets $\bT_i$, $\bL_i$, and $\bN_i$. Using these bit sequences we can map the symbol of a label in the global alphabet to the symbol of the same label in the effective alphabet\footnote{An effective alphabet of a sequence $S$ contains only symbols that occur in $S$ at least once.} used in one of subsets. When a label $a$ is deleted, we mark $SN[a]$ as free. When a new label $a'$ is inserted, we set $SN[a']$ to a free slot in $SN$ (a list of free slots is maintained). When some subset, say $\bC_i$ is re-built, we also re-build the bit sequence $GC_i$. 

In order to list objects related to a label $a$, we first report all objects that are related to $SN[a]$ and stored in $\bC_0$. Then 
we visit all subsets $\bC_i$, $\bL_i$, and $\bT_i$ and report all objects related to $\ra_1(SN[a],GC_i)$, $\ra_1(SN[a],GL_i)$, and $\ra_1(SN[a],GT_i)$ respectively. We remark that a global symbol of a label can be mapped to a wrong symbol in the local effective alphabet. This can happen if some label $a'$ is removed and its slot in $SN[]$ is assigned to another label $a$ but the bitmap of 
say $GC_i$ is not yet re-built. In this case $\ra_1(SN[a],GC_i)$ will map $a$ to the symbol for the wrong label $a'$. But $a'$ can be removed only if all object-label pairs containing $a'$ are deleted; hence, all pairs $(i,a')$ in $\bC_i$ are marked as deleted and the query to $\bC_i$ will correctly report nothing.  We can report labels related to an object and tell whether a certain object is related to a certain label using a similar procedure. 
We visit $O(\log \log n)$ data structures in order to answer a query. In all data structures except for $\bC_0$, we spend $O(\log\log \sigma_l)$ time per reported datum. An existential query on $\bC_0$ takes $O((\log\log \sigma_l)^2)$ time; all other queries 
on $\bC_0$ take $O(1)$ time per reported datum.  Hence all queries are answered in $O(\log\log n\log\log \sigma_l)$ time per reported datum. A counting query takes $O(\log n/\log \log n)$ time in each subset. Hence, we can count objects related to a label or labels related to an object in $O(\log n)$ time. 

All bit sequences $D$ and $D_a$ in all subsets use $O((n/\tau)\log\tau)$ bits. Every string $S$ stored in a deletion-only data structure needs $|S|H_0(S)+o(|S|\log\sigma_l)$ bits.  Hence all strings $S$ use at most $nH+o(n\log\sigma_l)$ bits, where $H=\sum_{1\le a\le \sigma_l}\frac{n_i}{n}\log\frac{n}{n_i}$. Bit sequences $GC_i$, $GL_i$, and $GT_i$ use $O(\sigma_l\tau)=o(n\log\sigma_l)$ bits.
Now we consider the space usage of bit sequences $N$ stored in deletion-only data structures. Let $m_i$ denote the number of pairs in a data structure $\bT_i$. $N$ consists of~$m_i$ $1$'s and $t$ $0$'s. If $m_i>t$, then the bit sequence $N$ stored as a part of~$\bT_i$ uses $m_i\log\frac{m_i+t}{m_i}=O(m_i)$ bits. If $t\ge m_i$, $N$ uses $O(m_i\log\tau)$ bits because $m_i=\Theta(n/\tau)$.  Hence all $N$ stored in all $\bT_i$ use $O(n\log\tau)$ bits. In our data structure we set $\tau=\log\log n$. If $\sigma_l=\Omega(\log^{1/4} n)$, $O(n\log\tau)=o(n\log\sigma_l)$. Otherwise $t=\Omega(n/\log n)$ because $n\le t\cdot\sigma_l$; if $t=\Omega(n/\log n)$, $O(n\log\tau)= o(t\log t)$. Data structures that are re-built at any moment of time contain $O(n/\tau)$ elements and use $O(\frac{n}{\tau}\log\sigma_l)=o(n\log\sigma_l)$ bits. Extra space that we need to store elements marked as deleted is bounded by $o(n\log\sigma_l)$; this can be shown in the same way as in Section~\ref{sec:worst}.
\begin{theorem}
\label{theor:binrel}
A dynamic binary relation that consists of~$n$ pairs relating $t$ objects to $\sigma_l$ labels can be stored in $nH + o(n\log\sigma_l) + o(t\log t) +O(t+n+\sigma_l\log n)$ bits where 
$H=\sum_{1\le a\le \sigma_l}\frac{n_a}{n}\log\frac{n}{n_a}$ and $n_a$ is the number of objects related to a label $a$. We can determine whether an object and a label are related in $O(\log\log \sigma_l\log \log n)$ time and report all objects related to a label (resp.\ all labels related to an object) in $O((k+1)\log\log\sigma_l\log \log n)$ time, where $k$ is the number of reported items. We can count objects related to a label or labels related to an object in $O(\log n)$ time. 
Updates are supported in $O(\log^{\eps} n)$ time. 
\end{theorem}

Directed graph is a frequently studied instance of a binary relation. In this case both the set of labels and the set of objects are identical with the set of graph nodes. There is an edge from a node $u$ to a node $v$ if the object $u$ is related to the label $v$. 
\begin{theorem}
\label{theor:graph}
A dynamic directed graph that consists of~$\sigma_l$ nodes and $n\ge\sigma_l$ edges can be stored in $nH + o(n\log\sigma_l) +O(n+\sigma_l\log n)$ bits where 
$H=\sum_{1\le a\le \sigma_l}\frac{n_a}{n}\log\frac{n}{n_a}$ and $n_a$ is the number of outgoing edges from node $a$. We can determine if there is an edge from one node to another one in $O(\log\log \sigma_l\log \log n)$ time and report all neighbors (resp.\ reverse neighbors) of a node in $O((k+1)\log\log\sigma_l\log \log n)$ time, where $k$ is the number of reported nodes. We can count neighbors or reverse neighbors of a node in $O(\log n)$ time. Updates are supported in $O(\log^{\eps} n)$ time. 
\end{theorem}

%\section{Additional Figures}

\section{Conclusions}
In this paper we described a general framework for transforming static compressed indexes into dynamic ones. We showed that, using our framework, we can achieve the same or almost the same space and time complexity for dynamic indexes as was previously obtained by static indexes. Our framework is applicable to a broad range of static indexes that includes a vast majority of currently known results in this area. Thus, using our techniques, we can easily modify almost any compressed static index, so that insertions and deletions of documents are supported. It will likely be possible to apply our framework to static indexes that will be obtained in the future. 
Our approach also significantly reduces the cost of basic queries in compact representations of dynamic graphs and binary relations. We expect that our ideas can be applied to the design  of other compressed data structures.

\paragraph{Acknowledgments} The authors wish to thank Djamal Belazzougui for clarifying the construction time of the static index in~\cite{BelazzouguiN11} and Gonzalo Navarro for explaining some technical details of dynamic indexes used in~\cite{MN08}.

\bibliographystyle{abbrv}
\bibliography{dynrank}

%\bibliography{dynrank}

\newpage
\appendix
\renewcommand\thesection{A.\arabic{section}}
\section{Reporting $1$-Bits  in a Bit Vector}\label{sec:fastbit}
We show how to store a bit vector with a small number of zeros 
in small space, so that all $1$-values in an arbitrary range can be reported in optimal time. This result is used by our method that transforms a static index into an index that supports deletions. We start by describing an $O(n)$-bit data structure.  Then we show how space usage can be reduced to $O((n\log \tau)/\tau)$
\begin{lemma}
   \label{lemma:fastbit}
  There exists an $O(n)$-bit data structure that supports the following operations on a bit vector $B$ of size $n$: 
  (i) $zero(i)$ sets $B[i]=0$ (ii) $report(s,e)$ enumerates all $j$ such that $s\le j\le e$ and $B[j]=1$. Operation $zero(i)$ is supported in $O(\log^{\eps}n)$ time and 
    a query $report(s,e)$ is answered in $O(k)$ time, where $k$ is the number of output bit positions. 
\end{lemma}
\begin{proof}
We divide the vector $B$ into words $W_1$, $\ldots$,$W_{\ceil{|B_i|/\log n}}$ of~$\log n$ bits. 
We say that a word $W_t$ is non-empty if at least one bit in $W_t$ is set to $1$. We store the indices of all non-empty words in a data structure 
that supports range reporting queries in $O(k)$ time, where $k$ is the number of reported elements, and updates in $O(\log^{\eps}n)$ time~\cite{MPP05}. 
For every word $W_i$ we can find the rightmost bit set to $1$ before the given position $p$ or determine that there is no bit set to $1$ to the right of~$p$ 
in $O(1)$ time. This can be done by consulting a universal look-up table of size $o(n)$ bits. To report positions of all $1$-bits in $B[s..e]$, we find all 
non-empty words whose indices are in the range $[\ceil{s/\log n},\floor{e/\log n}]$. For every  such word, we output the positions of all $1$-bits. Finally, we also examine the words $W_{\floor{s/\log n}}$ and $W_{\ceil{e/\log n}}$ and report positions of~$1$-bits in these two words that are in $B[s..e]$. 
The total query time is  $O(k)$.  
Operation $zero(i)$ is implemented by setting the bit $i-\floor{i/\log n}\log n$ 
in the word $W_{\ceil{i/\log n}}$ to $0$. If $W_{\ceil{i/\log n}}$ becomes empty, we remove 
$\ceil{i/\log n}$ from the range reporting data structure. 
\end{proof}

\begin{lemma}
   \label{lemma:fastbit2}
  Let $B$ be a bit vector of size $n$ with at most $O(\frac{n}{\tau})$ zero values for $\tau=O(\log n/\log \log n)$.  $B$ can be stored in $O(n\frac{\log\tau}{\tau})$-bit data structure that supports the following operations on  $B$: 
  (i) $zero(i)$ sets $B[i]=0$ (ii) $report(s,e)$ enumerates all $j$ such that $s\le j\le e$ and $B[j]=1$. Operation $zero(i)$ is supported in $O(\log^{\eps}n)$ time and  a query $report(s,e)$ is answered in $O(k)$ time, where $k$ is the number of output bit positions. 
\end{lemma}
\begin{proof}
We divide $B$ into words $W_i$ of~$\tau$ bits. Indices 
of non-empty words are stored in the the data structure $B'$,  implemented as in Lemma~\ref{lemma:fastbit}. Every word $W_i$ is represented as follows: we store the number of zeros 
in $W_{i}$ using $O(\log\tau)$ bits. A subword with $f$ zeros,
where $0\le f\le \tau$, is encoded using $f(\log\tau)$
bits by specifying positions of~$0$-bits.  For every word $W_i$ we can find the rightmost bit set to $1$ before the given position $p$ or determine that there is no bit set to $1$ to the right of~$p$ in $O(1)$ time. This can be done by consulting a universal look-up table of size $o(n)$ bits. 
Query processing is very similar to Lemma~\ref{lemma:fastbit}.
To report $1$-bits in $B[s..e]$, we find all 
non-empty words whose indices are in the range $[\ceil{l/\tau},\floor{r/\tau}]$. For every  such word, we output the positions of all $1$-bits. Finally, we also examine the words $W_{\floor{s/\tau}}$ and $W_{\ceil{e/\tau}}$ and report positions of~$1$-bits in these two words that are in $B[s..e]$. 

Operation $zero(i)$ is implemented by setting the corresponding bit in some word $W_i$ to $0$ and changing the word encoding.
If $W_i$ becomes empty, the $i$-th bit in $B'$ is set to $0$.
Issues related to memory management can be resolved as in \cite{NavarroN13}.

We need $O(n/\tau)$ bits to store the data structure $B'$  for non-empty words. Let $n_f$ denote the number of words with $f$ zero values.  All words 
$W_i$ need  $\sum_{f=1}^{\tau} n_f\cdot f\cdot \log\tau=\log\tau \sum n_f\cdot f=O((n/\tau)\log \tau)$ because $\sum n_f\cdot f =O(n/\tau)$. 
\end{proof}

\longver{
\section{Dynamic Document Collection in $O(n\log n)$ bits}
    \label{sec:uncompr}

A generalized suffix tree is a compact trie that contains all suffixes of all documents. Trie edges are labeled with strings and leaves correspond to suffixes of documents. Each document ends with a unique special symbol $\$_i$, hence all suffixes are unique. Every internal node has at least two children. Let $path(v)$ denote the string obtained by concatenating labels on the path from the root to a node $v$. The \emph{locus} of a string $P$ is the highest node $v$ such that $P$ is a prefix of $path(v)$. For every leaf $u$, $path(u)$ corresponds to a suffix.  Each occurrence of $P$ corresponds to a unique leaf that descends from the locus node of $P$. See e.g.,~\cite{Gus97} for a more detailed description of suffix tree.

We keep the collection $\cC_0$ in a generalized suffix tree (GST) augmented with suffix links.  
 A suffix link for a node $u$ labelled with a string $aX$ points to a node $v$ labelled with a string $X$. We use the algorithm of McCreight for inserting a new string into a suffix tree. When a new text $T$ is inserted we find the position of the string $T$ in the GST. Then we insert a leaf $u_l$ labelled with the suffix $T[1..|T|]$; if necessary, we also insert a parent node of~$u_l$ into the GST.  Then we follow the suffix link in the lowest ``old'' ancestor of~$u_l$ (i.e., the lowest node on the path to $u_l$ that existed before the insertion of~$T$ started). If this link points to some node $v$, we descend from $v$ as far as possible. Then we insert a new leaf $v_l$ corresponding to $T[2..|T|]$ and possibly 
the parent of~$v_l$. This procedure continues until all suffixes of~$T$ are inserted. Deletions are symmetric. The number of traversed edges and inserted nodes is 
$O(|T|)$. Every insertion of a new node takes $O(1)$ time.
 
To navigate in the suffix tree, we need a data structure $D(u)$ in each internal node $u$. For every child $u_i$ of~$u$, $D(u)$ contains the first character 
$a_i$ of the edge label $l(u,u_i)$, where $l(v,w)$ denotes an edge between nodes $v$ and $w$. For every alphabet symbol $a$, $D(u)$ returns a pointer to the edge $l(u,u_i)$ whose label starts with $a$ or reports that such edge does not exist. 
We can implement $D(u)$ in such way that queries and updates take $O(1)$ time.
If the alphabet size $\sigma$ is poly-logarithmic in $n$, we can use the data structure of Fredman and Willard~\cite{FW94}. 
If the alphabet size is large, $\sigma=\log^{\omega(1)}n$, we use the dynamic hashing to keep all children of a node $u$. In the latter case, the update time is randomized. If the alphabet size is large and  updates are using a deterministic algorithm, then 
we  implement $D(u)$ as an exponential tree~\cite{AnderssonT07};
in this case an appropriate child $u_i$ of~$u$ is found in $O((\log\log\sigma)^2)$ time.

Occurrences of a pattern $P$ are reported using the standard suffix tree procedure.
We traverse the search path for a pattern $P$ starting at the root node and choosing the child $u_i$ of the current node $u$ that is labelled with a prefix of~$P$ until the locus of~$P$ is found or the search cannot continue. In each 
visited node $u$ we search for $p_i$ in $D(u)$, where $p_i$ is the next unprocessed symbol in $P$. If $u_i$ is labelled with a prefix of~$P$, the search continues in $u_i$. Otherwise the search ends on the edge from $u$ to $u_i$. 
When the locus of a pattern $P$ is found, we can report all occurrences of~$P$ in $O(1)$ time per occurrence. 

We can also modify our data structure so that the locus of~$P$ is found in $O(|P|/\log_{\sigma}n(\log\log \sigma)^2 +\log n)$ time~\cite{NavarroN13arx}.  If the update procedure uses randomization, then the locus of~$P$ can be  found in $O(|P|/\log_{\sigma}n+\log^{\eps}n)$ time.

\section{Maintaining the Sizes of Sub-Collections after Updates in Transformation~2}
\label{sec:updatesapp}
We show here how to maintain the invariant $n_f=\Theta(n)$.
If $n\ge 2n_f$ after a document insertion, we set $n_f=n$. Maximal sizes $\max_i$ of subcollections $\cC_i$ are changed accordingly. All top sub-collections $\cT_i$ that contain less than $n_f/\tau$ symbols are merged into new collections $\cT'_l$ of total size between $n_f/\tau$ and $2n_f/\tau$ symbols. During the next $n_f/\tau$ symbol updates (that is, insertions and deletions of texts of total size $n_f/\tau$), we construct new collections $\cT'_l$. 

$\cT_i$ that must be re-built are processed one-by-one.
Since at any moment only one $\cT'_i$ is constructed, this process 
needs $O(nw(n)/\tau)$ bits of workspace. 

If $n\le n_f/2$ after a document deletion, we set $n_f=n/2$. All $\cT_i$ that contain more than one document and satisfy $|\cT_i|\ge n_f/\tau$ are split into two subcollections $\cT'_i$. Each document $T$ from $\cT_i$, such that $|T|\ge n_f/\tau$ is assigned to its own one-document  collection $\cT'_j$.  Other documents are assigned to collections of size between $n_f/\tau$ and $n_f/2\tau$ symbols. We also move all documents from collections $\cC_j$, $j=0,\ldots, r$, to one or two new collections $\cT'_{i_1}$ and $\cT'_{i_2}$, such that $\cT'_{i_1}$, $\cT'_{i_2}$ contain between $n_f/\tau$ and $n_f/2\tau$ symbols.
 At any moment only one new collection $\cT'_i$ is constructed. Hence this process also needs  $O(nw(n)/\tau)$ bits of workspace. We can schedule the rebuilding in such way that all $\cT'_i$ are finished after the following $n_f$  symbol updates.  

We also start the re-building process every time when a one-document collection $\cT_i$ is inserted or deleted. In this case we update the value of~$n_f$ and re-build the subcollections as described above (if there is another process for replacing $\cT_i$ with $\cT_i'$ that currently runs in the background, then this process is terminated). Since $\cT_i$ contains (resp. contained) a document $T$ of size $\Omega(n/\tau)$, re-building subcollections takes $O(|T|\tau\cdot u(n))$ time. \footnote{We assume here that when a new document $T$ is inserted, then $T$ is  stored in uncompressed form. Hence, the procedure that constructs 
a new one-document collection $\cT_i$ can use $O(|\cT_i|\log\sigma)$ bits of space. Alternatively we can assume that very big documents are split into several parts of at most $n/\tau$ symbols and each part is kept in a separate $\cT_i$.}
Hence, $n_f=\Theta(n)$ at any time.

\section{Dynamic Transformation with Lower Update Cost}
\label{sec:fasteramortized}
\begin{transform}
\label{trans:trans2}
Suppose that there exists a static $(u(n), w(n))$-constructible index $\cI_s$ that uses $|S|\phi(S)$ space for any document collection $S$. Then there exists a dynamic index $\cI_d$ that uses $|S|\phi(S)+O(|S|(\frac{\log\tau+\log \sigma}{\tau}+w(n)))$ space for any parameter $\tau=O(\log n/\log \log n)$; $\cI_d$ supports insertions and deletions of documents  in  $O(u(n)\log \log n)$ time per symbol and $O(u(n)\cdot \tau+ \tsa+ \log^{\eps}n)$ time per symbol respectively. Update times are amortized. The asymptotic cost of range-finding increases by factor $O(\log \log n)$;
the costs of  extracting and locating are the same in $\cI_s$ and $\cI_d$. 
\end{transform}
We divide the document collection $\cC$ into sub-collections $\cC_1,\ldots, \cC_r$ such that $|\cC_i|\le \max_i$ and $\max_i=2(n/\log^2n) 2^{i}n$ for $i=0,1,\ldots,r$.
Thus the number of sub-collections is $r=O(\log \log n)$. All collections $\cC_i$ are organized, queried, and updated in exactly the same way as in Transformation~\ref{trans:trans1}. 
Since we must query $O(\log \log n)$ 
sub-collections, the time to answer a range-finding query grows by $O(\log \log n)$ factor.
Deletion time is the same as in Transformation~\ref{trans:trans1} because the same deletion-only indices for sub-collections are used.  
Analysis of insertion costs is similar to Transformation~\ref{trans:trans1}. Between two global rebuilds 
every text is inserted into each sub-collection at most once. 
When  a sub-collection $\cC_i$ is re-built, we insert $\Omega(|\cC_i|)$ new symbols into $\cC_i$. Hence, re-building a collection incurs an amortized cost of~$O(u(n))$ on every new symbol in $\cC_i$. Thus the total amortized cost of an insertion is $O(u(n)\log\log n)$.

\section{Analysis of Space Usage}
\label{sec:spacean}
In this Section we show that the space overhead
caused by keeping deleted symbols is bounded. 
Suppose that $n/\tau$ symbols from some documents are marked as deleted in a collection $\cC$. Let $\cC'$ denote the collection 
$\cC$ without deleted documents. In this section we consider the case when the space usage of $\cC$ is bounded by $nH_k+o(n)$ for some $k\ge 1$.

A context $c_i$ is an arbitrary sequence of length $k$ over an alphabet $\sigma$; for simplicity we identify a context $c_i$ by its index $i$ where $i\in [1,\sigma^k]$. 
Let $f_{a,i}$ and $f'_{a,i}$ denote the number of times the symbol $a$ occurs in  the context $i$ in $\cC$ and $\cC'$ respectively. Let $n_i=\sum_{a}f_{a,i}$ and $n'_{i}=\sum_a f'_{a,i}$.  
The $k$-th order empirical entropy of~$\cC$ is defined as $\sum_{c_i\in \Sigma^k}\sum_{a\in\Sigma} f_{a,i}\log\frac{n_i}{f_{a,i}}$.

We need $F_1=\sum_i\sum_a f'_{a,i}\log\frac{n_{i}}{f_{a,i}}$ bits to keep all deleted symbols. We express  $\log\frac{n_i}{f_{a,i}}=\log\frac{n_i}{n'_i}+\log\frac{n'_i}{f'_{a,i}}+\log\frac{f'_{a,i}}{f_{a,i}}< \log\frac{n_i}{n'_i}+\log\frac{n'_i}{f'_{a,i}}$. 
Furthermore $\sum_i\sum_a f'_{a,i}\log\frac{n'_i}{f'_{a,i}}\le 
\frac{n}{\tau}\log\sigma$. 
We can also show that $\sum_i n'_i\log\frac{n_i}{n'_i}= o(n)$. 
All contexts $i$ are divided into three sets. Let $I_1$ contain all context indices, such that $n_i\ge n'_i\log^2n$. 
For all $i\in I_2$, $n_i\log^2 n> n_i'\ge n_i(\log\log n)^2$.
For all $i\in I_3$, $n_i(\log\log n)^2> n_i'$. 
Since $\sum n_i=O(n)$, $\sum_{i\in I_1} n'_i\log\frac{n_i}{n'_i}+\sum_{i\in I_2} n'_i\log\frac{n_i}{n'_i}=O(n)(\frac{1}{\log n}+\frac{1}{\log\log n})=o(n)$.
Since $\sum_i n'_i=O(n/\tau)$, $\sum_{i\in I_3} n'_i\log\frac{n_i}{n'_i}=O(\frac{n}{\tau}\log^{(3)}n)=o(n)$ for $\tau=\Omega(\log^{(3)}n)$.  
Hence $F_1=(n/\tau)\log\sigma +o(n)$. 

The contexts of most symbols in $\cC'$ are the same as in $\cC$. Only first $k\le \log_{\sigma}n/2$ symbols in each document can change context (because the previous document was deleted). The total number of such symbols is bounded by $\rho\cdot k$. These 
symbols are encoded in $O(\rho\log n)+o(n)$ bits. Contexts of remaining symbols in $\cC'$ remain unchanged. 
The space consumed by other (not deleted) symbols can be still slightly higher than optimal. Let $\of_{a,i}=f_{a,i}-f'_{a,i}$ and 
$\on_{a,i}=\sum_a \of_{a,i}=n_{i}-n'_i$. For simplicity we ignore symbols that changed contexts. All undeleted symbols use $E_u=\sum_i\sum_a \of_{a,i}\frac{n_i}{f_{a,i}}$ bits. Optimal compression of the same sequence would use $E_o=\sum_i\sum_a \of_{a,i}\frac{\on_i}{\of_{a,i}}$. $F_2=E_1-E_2\le \sum_i\sum_a \of_{a,i}\log\frac{n_i}{\on_i}=\sum_i \on_i \log\frac{n_i}{\on_i}=O(n)$. 
Thus the total additional space is $F_1+F_2=O(n\frac{\log\sigma}{\tau})+o(n\log\sigma)$.

\section{Construction Times of Static Indexes}
\label{sec:constimeapp}
% \paragraph{Small Alphabets}
% The static index of  Ferragina et al.~\cite{FMMN07}  for alphabets of poly-logarithmic size is $O(\log n/\log \log n,o(1))$-constructible. Their index is essentially a data structure that supports rank and select operations on  a BWT-transformed text. We can obtain a Burrows-Wheeler transform 
% (BWT) of a string $T$ in $O(|T|\log n/\log \log n)$ time~\cite{NavarroN13} using $o(|T|)$ additional bits.
% When the BWT-transform is available, we can construct 
% a data structure that supports rank, select, and access in $O(|T|)$ time using $O(|T|)$ bits of workspace; details will be given in the full version. Thus the index from\cite{FMMN07} is $(\log n/\log\log n, o(1))$-constructible. 
\paragraph{Arbitrarily Large Alphabets}
% For arbitrarily large alphabets we used the index of Belazzougui and Navarro~\cite{BelazzouguiN11} that achieves $\trange=O(|P|)$, $\textract=O(s+\ell)$, and $\tlocate=O(s)$. It can be shown that the index described in~\cite{BelazzouguiN11} is $(\log^{\eps}n,\log\sigma)$-constructible. We apply Transformation~\ref{trans:trans1worst} with $\tau=\log\log n$ to this index. 
% We obtain a dynamic index that uses $nH_k+O(n\frac{\log\sigma}{\log\log n})+O(n\frac{\log n}{s})$ bits and achieves $\trange=O(|P|\log\log n)$, $\textract=O(s+\ell)$, and $\tlocate=O(s)$. Most previously described static indexes~\cite{BHMR07} have the same redundancy with respect to the alphabet size 
% and the parameter $s$, $O(n\frac{\log\sigma}{\log\log n})+O(n\frac{\log n}{s})$. Two recent indexes\cite{BGNN10}, \cite{BelazzouguiN11} have redundancy 
% $O(n\frac{H_k}{\log\log n})+O(n\frac{\log n}{s})$. Our dynamic index supports insertions and deletions in $O(|T|\log^{\eps}n)$ time and $O(|T|(\log^{\eps}n+s))$ time 
% respectively.

It can be shown that the index of Belazzougui and Navarro~\cite{BelazzouguiN11} is $(\log^{\eps}n,\log\sigma)$-constructible.
  This index consists of three components. First, a BWT transform is applied to the source text. Then a data structure of Barbay et al.~\cite{BGNN10} for the BWT-transformed text is created; this data structure supports select queries in $O(1)$ time. Second, a compressed suffix tree for the source text is created. Third, we keep w-links on the compressed suffix tree. A w-link for an alphabet symbol $a$ points from a node $u$ that is labelled with a suffix $X$ to a node or a position in a tree that is labelled with a suffix $aX$; if $aX$ does not occur, then the link for $a$ and $u$ does not exist. W-links are implemented using a collection of monotone minimum perfect hash functions (mmphf)~\cite{BelazzouguiBPV09}. 

The index from~\cite{BelazzouguiN11} can be constructed as follows.  First, we construct a compressed suffix tree in $O(n\log^{\eps}n)$ time using  $O(n \log\sigma)$ bits of extra space by employing the algorithm described in~\cite{HSS09}. 
Then we traverse the tree and produce mmphf in $O(n)$ randomized time. 
Next we obtain the BWT transform of~$T$; this step takes $O(n\log\sigma)$ extra bits and $O(n)$ time. Finally, we construct the data structure from~\cite{BGNN10}. Our method for constructing the data structure is as follows: Let $T^b$ denote the BWT-transformed sequence.
We split $T^b$  into chunks $C_j$, such that each chunk but the last consists of~$\sigma^2$ symbols and the last $C_l$ consists of at most $\sigma^2$ symbols. Then the data structure of~\cite{BGNN10} is constructed for each chunk. The symbols of~$C_i$ are distributed among $O(\log \sigma)$ groups $G_s$.
Each symbols in $G_i$ occurs at least $2^i$ and at most $2^{i+1}$ times for 
$i=1,2,\ldots, \log |C_i|$. This step takes linear time and $O(\sigma\log \sigma)$ extra bits. 
Let $C_{s,i}$ denote the subsequence of~$C_s$ induced by symbols of~$G_i$; 
let $C_s(G)$ denote the sequence that specifies the group index for every symbol of~$C_s$. We replace $C_s$ with $C_s(G)$ and subsequences $C_{s,i}$. 
Data structures supporting rank, select queries are stored for $C_s(G)$ and 
all $C_{s,i}$.  The data structure for $C_s(G)$ is implemented as described 
in~\cite{FMMN07}; the data structures for $C_{s,i}$ are implemented as described by  Golynski et al.~\cite{GMR06}.
Both data structures can be constructed in linear time using $o(|C_{s,i}|)$ additional bits.
If select queries on each chunk can be answered in $O(1)$ time, we can also answer select queries on $T^b$ using  $O(n)$ additional bits. The method is based on keeping a bit vector $B_a=1^{j_1}01^{j_2}\ldots 1^{j_f}0$ for every symbol $a$, where $f$ is the number of 
chunks and $j_i$ is the number of times $a$ occurs in the $i$-th chunk $C_i$. We create a data structure that answers rank and select 
queries on $B$. Then, we can identify the chunk that contains the $l$-th occurrence of~$a$ by answering a query $rank_0(select_1(l,B_a),B_a)$. Then we identify the position of~$l$-th occurrence of~$a$ by a query $select_a(l-l',C_h)$, where 
$l'=rank_1(h-1,B_a)$. Data structures for a chunk $C_i$ can be constructed in linear time using $O(|C_i|\log\sigma)$ bits of workspace.

\paragraph{Index of Barbay et al.   \cite{BGNN10}}
This index is a part of the data structure of  Belazzougui and Navarro~\cite{BelazzouguiN11}. Hence it is also $(\log^{\eps}n,\log\sigma)$-constructible. Unlike the structure in~\cite{BelazzouguiN11} the index described in~\cite{BGNN10} can be constructed by a deterministic algorithm. 

\paragraph{$O(n\log\sigma)$-bit Index}
The index of Grossi and Vitter~\cite{GrossiV05} is also $(\log^{\eps}n,\log\sigma)$-constructible. Their index consists of the compressed suffix array $CSA$ and functions 
$\Psi^k(i)=SA^{-1}[SA[i]+k]$ for $k=1,\ldots,\log^{i\eps}n,\ldots$ and $i=0,1,\ldots (1/\eps)$. Using the algorithm of Hon et al~\cite{HSS09}, we can construct 
$CSA$ and $\Psi^k$ in $O(n\log \log \sigma)$ time using $O(n\log \sigma)$ bits. To speed up the range finding, Grossi and Vitter store a series of suffix trees
for subsequences of the suffix array. The top level tree is a compressed trie over $s_1=n/\log_{\sigma}n$ suffixes $SA[1]$, $SA[1+\log_{\sigma}n]$, $\ldots$. 
On the next level, we consider each subarray $SA_h=SA[(h-1)\log_{\sigma}n+1..h\log_{\sigma}n]$. We select every $\log_{\sigma}^{\eps/2}n$-th suffix from $SA_h$ and construct 
a suffix tree for this set of suffixes. On the next level, we consider subarrays of size $\log_{\sigma}^{1-\eps/2}n$, select every $\log_{\sigma}^{\eps/2}n$-th suffix 
and construct a suffix tree for the resulting subset. This subdivision continues untill the size of the subarray is equal to $\log_{\sigma}^{\eps}n$. These suffix tree can be constructed in $O(n\log^{\eps}n)$ time and $O(n\log\sigma)$ bits: the total number of leaves in all suffix trees is $o(n)$ and a suffix tree 
for $m$ suffixes can be constructed in $O(m\log^{\eps}n)$ time~\cite{HSS09}.  The search for a range of the suffix array that corresponds to the query pattern is described in\cite{GrossiV05}. Thus the index from \cite{GrossiV05} can be constructed in 
$O(n\log^{\eps}n)$ time using $O(n\log\sigma)$ space. 

% We apply Transformation~\ref{trans:trans1worst} with $\tau=1/\delta$ for a constant $\delta$. The resulting dynamic index uses $O(n\log\sigma(1+1/\delta))=O(n\log\sigma)$ bits and has the same query times as the static index, $\tlocate=O(\log^{\eps}n)$, $\trange=O(|P|/\log_{\sigma}n +\log^{\eps}n)$ and $\textract=O(\ell/\log_{\sigma}n)$.  

}

\end{document}